\documentclass[journal]{IEEEtran}
\usepackage{graphics}
\usepackage{graphicx}
\usepackage{amsfonts}
\usepackage{multirow}
\usepackage{epstopdf}
\usepackage{color}
\usepackage{float}
\usepackage[font=small]{caption}
\usepackage{bm}
\usepackage{algorithm}
\usepackage{algorithmic}
\usepackage{amsmath}
\usepackage{pifont}
\usepackage[utf8]{inputenc}
\usepackage{mathtools}
\usepackage{amssymb}
\usepackage{amsthm}
\usepackage{lipsum}% http://ctan.org/pkg/lipsum
\usepackage{bbm}
\usepackage{enumitem}   
\usepackage{subcaption}
\usepackage{cite}
\usepackage{steinmetz} % For Phase Symbol
\usepackage{placeins} % For FloatBarrier

\usepackage[table]{xcolor}
\newcolumntype{s}{>{\columncolor[HTML]{FE6F5E}} c}
\newcolumntype{t}{>{\columncolor[HTML]{5D8AA8}} c}

\usepackage{booktabs}% http://ctan.org/pkg/booktabs

\newtheorem{theorem}{Theorem}

\newtheorem{assumption}{A}

\theoremstyle{definition}
\newtheorem{definition}{Definition}

\ifCLASSINFOpdf
\else
\fi

\begin{document}

\title{Sparse Convolutional Beamforming \\ for Ultrasound Imaging}

\author{Regev~Cohen,~\IEEEmembership{Student,~IEEE,}
        Yonina~C.~Eldar,~\IEEEmembership{Fellow,~IEEE}
\thanks{This research was supported by the I-CORE Program of the Planning and Budgeting Committee and The Israel Science Foundation 1802/12.}}

\maketitle

\begin{abstract}
The standard technique used by commercial medical ultrasound systems to form B-mode images is delay and sum (DAS) beamforming. However, DAS often results in limited image resolution and contrast, which are governed by the center frequency and the aperture size of the ultrasound transducer. A large number of elements leads to improved resolution but at the same time increases the data size and the system cost due to the receive electronics required for each element. Therefore, reducing the number of receiving channels while producing high quality images is of great importance. In this paper, we introduce a nonlinear beamformer called COnvolutional Beamforming Algorithm (COBA), which achieves significant improvement of lateral resolution and contrast. In addition, it can be implemented efficiently using the fast Fourier transform. Based on the COBA concept, we next present two sparse beamformers with closed form expressions for the sensor locations, which result in the same beam pattern as DAS and COBA while using far fewer array elements. Optimization of the number of elements shows that they require a minimal number of elements which is on the order of the square root of the number used by DAS. The performance of the proposed methods is tested and validated using simulated data, phantom scans and \textit{in vivo} cardiac data. The results demonstrate that COBA outperforms DAS in terms of resolution and contrast and that the suggested beamformers offer a sizable element reduction while generating images with an equivalent or improved quality in comparison to DAS.   
\end{abstract}


\section{Introduction}
\label{sec:intro}

\IEEEPARstart{U}{ltrasound} imaging is one of the most common medical imaging modalities, allowing for non-invasive investigation of anatomical structures and blood flow. Cardiac, abdominal, fetal and breast imaging are some of the applications where it is extensively used as a diagnostic tool.

In a conventional scanning process, short acoustic pulses are transmitted along a narrow beam from an array of transducer elements. During their propagation echoes are scattered by acoustic impedance perturbations in the tissue, and detected by the array elements. The backscattered radio-frequency (RF) signals are then processed in a way referred to as beamforming to create a line in the image. The beamformer is designed to focus and steer the ultrasound transducer towards a desired direction or point in space. The main goal of the beamformer is to generate a beam pattern with a narrow main lobe and low side lobes \cite{thomenius1996evolution}. The beam main-lobe width dictates the system resolution, while the side-lobe level controls contrast so that the beam properties have a great impact on image quality \cite{ranganathan2003novel}.

In medical ultrasound imaging, the standard beamformer is delay and sum (DAS) \cite{thomenius1996evolution,karaman1995synthetic} which consists of delaying and weighting the reflected echoes before summing them. While its simplicity and real-time capabilities make DAS widely used in ultrasound scanners, it exhibits limited imaging resolution and contrast \cite{jensen2006synthetic}. Increasing the number of elements, while keeping the array pitch below half a wavelength to avoid grating lobes \cite{lockwood1998real}, results in enhanced resolution. However, this increases channel data size and the system cost due to the receive electronics required for each element. Therefore, reducing the number of receiving channels while producing high quality images is of great importance.  

\subsection{Related Work}
Considering a full array, several methods to improve image quality have been proposed. Adaptive beamformers improve resolution without sacrificing contrast by dynamically changing the receive aperture weights based on the received data statistics \cite{van1997beamforming}. The most common is Capon/minimum variance (MV) beamforming \cite{viola2005adaptive} which offers better contrast and resolution than DAS. However, it is difficult to apply in real-time due to the calculation of a covariance matrix and its inverse at each time instant. Its application to ultrasound imaging was studied extensively over the last decade and many improved versions of MV with reduced complexity have been proposed \cite{chen2013improved,kim2014fast,synnevag2007adaptive,asl2009minimum}. Nilsen \textit{et. al.} suggest a beamspace adaptive beamformer, BS-Capon, based on orthogonal beams formed in different directions \cite{nilsen2009beamspace}. Jensen \textit{et. al.} developed an adaptive beamformer called multi-beam Capon that is based on multibeam covariance matrices \cite{jensen2012approach}. Using similar concepts, Jensen and Austeng proposed a method called iterative adaptive approach (IAA) \cite{jensen2014iterative,yardibi2008nonparametric}.

Other related techniques have been presented such as applying various finite impulse response (FIR) filters on each receive channel, instead of single apodization weights \cite{guenther2009robust}. Chernyakova \textit{et. al.} proposed a beamformer called iMAP where both the interference and the signal of interest are viewed as random variables and the beamformer output is the maximum-a-posteriori (MAP) estimator of the signal computed in an iterative fashion. An approach based on the spatial correlation of echo signals called SLSC has been suggested in \cite{lediju2011short,dahl2011lesion}. However, B-mode techniques aim at imaging the magnitude of the backscattered echoes, whereas SLSC attempts to calculate their spatial coherence. The authors in \cite{matrone2015delay,matrone2017high,matrone2017depth,matrone2015ultrasound,matrone2016ultrasound,matrone2018experimental,mozaffarzadeh2017double} presented a non-linear beamformer called FDMAS that is based on computing the auto-correlation of the RF signals. This approach leads to improved resolution and contrast at the expense of high computational load, resulting in slow runtime. 

Several studies investigate compressed sensing (CS \cite{eldar2012compressed,eldar2015sampling}) techniques for data reduction, based on the assumption that the ultrasound signal can be sparsely represented in an appropriate basis. Wagner \textit{et. al.} proposed a method for reducing the sampling rate \cite{wagner2012compressed} by treating ultrasound signals within the finite rate of innovation \cite{eldar2012compressed,eldar2015sampling} framework. Sub-Nyquist data acquisition from each transducer element and low-rate processing were presented in \cite{chernyakova2014fourier} and were later extended to plane-wave imaging \cite{chernyakova2018fourier}. Liebgott \textit{et. al.} studied the reconstruction performance of ultrasound signals in different bases \cite{liebgott2013pre}. In \cite{liu2017compressed}, the authors introduced a beamforming technique called compressed sensing based synthetic transmit aperture which increases the frame rate by transmitting a small number of randomly apodized plane-waves and uses CS reconstruction to recover the full channel data.       
None of the works above consider element reduction.

Possible approaches to reduce the number of receiving channels without compromising image quality include subaperture processors and microbeamformers \cite{larson19932}, whereby part of the beamformation is moved into the probe handle. However, this requires manufacturing expensive integrated circuits with high power consumption \cite{savord2003fully,fuller2005sonic,lee2004miniaturized}. Alternative strategies that have gained a lot of interest are based on using standard DAS with sparse arrays where some of the elements are removed, including deterministic designs such as vernier arrays and random designs \cite{davidsen1994two,brunke1997broad,yen2000sparse,austeng2002sparse,karaman2009minimally,diarra2013design,ramadas2014application,emmanuel2017validation,mitra2010general}. These works are concerned with designing a combined transmit/receive effective aperture mostly for 3D imaging, whereas we propose methods that can be applied in both active and passive settings.
% Moreover, random thinned arrays exhibit increased average side lobe levels.
Another approach is 2D row–column-addressed arrays for 3D imaging \cite{chen2011cmut,rasmussen20133d,rasmussen20133,rasmussen20153,christiansen20153}, in which every row and column in the array acts as one large element. However, this work is limited to 3D imaging and the use of large elements leads to a considerable increase in edge effects that limit image quality \cite{rasmussen20133d}. 
% Moreover, these solutions are mainly for 3D ultrasound imaging.        

\subsection{Contributions}

The main goal of this work it to reduce the number of receiving channels while preserving or improving the image quality in comparison to a DAS beamformer operating on the full array. To that end, we propose a new beamforming technique and present two deterministic designs of sparse arrays based on it.

We first introduce a non-linear beamformer referred to as convolutional beamforming algorithm (COBA), which is based on the convolution of the delayed RF signals prior to summation. COBA can be implemented efficiently using the fast Fourier transform (FFT), thus, making it suitable for real-time application. We analyze the beam pattern generated by COBA and show its relation to the sum co-array \cite{cohen2018optimized,hoctor1990unifying} which has twice the size of the physical aperture and triangle-shaped apodization. Consequently, COBA demonstrates significant improvement of lateral resolution and image contrast.

Then, we provide a definition of sparse arrays based on the sum co-array, which combined with COBA leads to two designs of sparse convolutional beamformers that require fewer receiving elements than DAS. The first technique, called sparse convolutional beamforming algorithm (SCOBA), utilizes significantly fewer elements while obtaining a beam pattern similar to that of DAS in terms of resolution. The second method, termed sparse convolutional beamforming algorithm with super-resolution (SCOBAR), offers increased resolution at the expense of a smaller, yet notable, channel reduction. We then describe how to apply apodization directly on the sum co-array in order to improve its contrast. Optimization of the sparse designs reveals that the minimal number of elements required to obtain the beam patterns achieved by SCOBA and SCOBAR are both proportional to $\sqrt{N}$, where $N$ is the number of channels in the fully populated array. Thus, these approaches offer sizable element reduction without compromising image quality.

Next, we use simulations of point-reflectors and an anechoic cyst to provide qualitative and quantitative assessments of image quality using the proposed beamformers. We show that COBA achieves significant improvement of resolution and contrast compared to DAS. In addition, SCOBA and SCOBAR demonstrate similar and enhanced performance with respect to DAS while operating with a low number of channels. These results are verified using phantom scans and \textit{in vivo} cardiac data, proving that the beamformers presented are suitable for clinical use in real-time scanners. 

The rest of the paper is organized as follows. In Section\,\ref{sec:model}, we describe the signal model and formulate our problem. Section \ref{sec:convbf} introduces the convolutional beamformer, applied to ultrasound image formation, and analyzes its beam pattern. We present and describe in detail sparse array designs in Section \ref{sec:sparseconvbf} and propose two beamformers that utilize fewer elements. We then derive the minimal number of channels required by both approaches. In Section \ref{sec:results}, the performance of the suggested techniques is evaluated using simulated and experimental data. Finally, Section \ref{sec:conclude} concludes the paper.            

\section{Array Theory and Problem Formulation}
\label{sec:model}

\subsection{Signal Model and Beam Pattern}
We consider a uniform linear array (ULA) comprised of $2N-1$ transducer elements aligned along the lateral axis $x$. The sensor locations $\{p_n\}$ are given by
\begin{equation}
p_n=(nd,z=0)\quad n=-(N-1),...,N-1,
\end{equation}
where $d$ is the spacing (pitch) between the centers of the individual elements and $z$ denotes the axial axis. Upon reception, an energy pulse is backscattered from a point in space $(r,\theta)$, propagates through the tissue at speed $c$ 
% At time $t \geq0$ its coordinates are $(x,z)=(ct\sin\theta,ct\cos\theta)$. a
and is received by all array elements at a time depending on their locations. We denote the signal received at the $0$th element by
\begin{equation}
f(t)=\tilde{f}(t)e^{jw_0t}.
\end{equation}  
Here $w_0$ is the transducer center frequency and $\tilde{f}(t)$ is the signal envelope. The sensors spatially sample the signal such that the signal $f_n(t)$ at the $n$th element is given by
\begin{equation}
f_n(t)=f(t-\tau_n)=\tilde{f}(t-\tau_n)e^{jw_0(t-\tau_n)},
\label{eq:fn}
\end{equation}
where $\tau_n$ is a time delay. To derive an expression for the delays, we introduce the following assumptions:
\begin{assumption}[Narrow-Band] 
The signal $f(t)$ is narrow-band, i.e., the bandwidth of the envelope is small
enough so that
\begin{equation}
\tilde{f}(t-\tau_n)\simeq \tilde{f}(t),\quad n=-(N-1),...,N-1.
\end{equation}
\label{ass:narrowband}
\end{assumption}
\begin{assumption}[Far-Field] 
The point $(r,\theta)$ is in the far-field region of the array, thus, the input signal impinging on the array is considered to be a plane-wave.
\label{ass:farfield}
\end{assumption}
\hspace{-0.4cm}Under the assumptions above, we can rewrite (\ref{eq:fn}) as
\begin{equation}
f_n(t)=\tilde{f}(t)e^{jw_0(t-\tau_n)},
\label{eq:newfn}
\end{equation}
where $\tau_n=\frac{d\sin\theta}{c}n$, i.e., the delays are approximated by phase shifts independent of $r$.

A beamformer processes each sensor output by a filter with impulse response $\tilde{g}_n(t)= g_n(t+\alpha_n)$ where 
\begin{equation}
\alpha_n=\frac{d\sin\theta_0}{c}n
\end{equation}
for a direction of interest  $-\frac{\pi}{2}\leq\theta_0\leq \frac{\pi}{2}$. Thus, the output of the $n$th element is
\begin{equation}
y_n(t)=\tilde{g}_n(t)\underset{t}{\ast} f_n(t)=g_n(t+\alpha_n)\underset{t}{\ast} f_n(t-\tau_n),
\label{eq:nthoutput}
\end{equation}
where $\underset{t}{\ast}$ denotes temporal convolution.
The beamformer then sums the outputs to obtain the array output
\begin{equation}
y(t)=\sum_{n=-(N-1)}^{N-1} y_n(t).
\label{eq:output}
\end{equation}
In the frequency domain, (\ref{eq:output}) may be expressed as
\begin{align}
\begin{split}
Y(\omega)&=\sum_{n=-(N-1)}^{N-1} G_n(\omega)F(\omega)e^{-j\omega(\tau_n-\alpha_n)} \\
&=F(\omega)\sum_{n=-(N-1)}^{N-1} G_n(\omega)e^{-j\omega(\tau_n-\alpha_n)},
\label{eq:outputfreq}
\end{split}
\end{align} 
where $Y(\omega)$, $F(\omega)$  and $G_n(\omega)$ are the temporal Fourier transforms of $y(t)$, $f(t)$ and $g_n(t)$ respectively.

To analyze the response of a beamformer to an input field, we assume the input to be a unity amplitude plane wave
\begin{equation}
f(t)=e^{jw_ot},
\label{eq:basisinput}
\end{equation} 
where $\tilde{f}(t)\equiv 1$.
Substituting (\ref{eq:basisinput}) into (\ref{eq:outputfreq}), we obtain
\begin{align}
\begin{split}
Y(\omega)&=\delta(\omega-\omega_0)\sum_{n=-(N-1)}^{N-1} G_n(\omega)e^{-j\omega(\tau_n-\alpha_n)} \\
&=\delta(\omega-\omega_0)\sum_{n=-(N-1)}^{N-1} G_n(\omega_0)e^{-j\omega_0(\tau_n-\alpha_n)},
\label{eq:freqresponse}
\end{split}
\end{align}
where $\delta(\omega)$ is the Dirac delta.
Given the explicit expressions for $\tau_n$ and $\alpha_n$, we can rewrite (\ref{eq:freqresponse}) as a function of $\theta$
\begin{equation}
Y(\omega,\theta)=\delta(\omega-\omega_0)
\sum_{n=-(N-1)}^{N-1} G_n(\omega_0)e^{-j\omega_0\frac{nd}{c}\left(\sin\theta-\sin\theta_0\right)}.
\label{eq:freqthetaresponse}
\end{equation}
The sum on the right hand side of (\ref{eq:freqthetaresponse}) is defined as the beam pattern of the beamformer
\begin{equation}
H(\theta)\triangleq
\sum_{n=-(N-1)}^{N-1} G_n(\omega_0)e^{-j\omega_0\frac{nd}{c}\left(\sin\theta-\sin\theta_0\right)}.
\end{equation}
For simplicity, we assume that $\theta_0=0$ which yields 
\begin{equation}
H(\theta)=\sum_{n=-(N-1)}^{N-1} G_n(\omega_0)e^{-j\omega_0\frac{nd\sin\theta}{c}}.
\label{eq:bp}
\end{equation}
The beam pattern represents the beamformer response to variations in the input field.

In standard delay and sum (DAS) beamforming we have
\begin{equation}
g_n(t)=w_n\delta(t),\quad n=-(N-1),...,N-1,
\end{equation}
where $w_n$ is the weight of the $n$th element. 
Thus,
\begin{equation}
H_\text{DAS}(\theta)=\sum_{n=-(N-1)}^{N-1} w_ne^{-j\omega_0\frac{nd\sin\theta}{c}}.
\label{eq:das}
\end{equation}
A plot of a beam pattern generated by a standard DAS beamformer is presented in Fig. \ref{fig:beampattern}.
The main lobe width of the beam pattern affects system resolution, while the peak side lobe level determines image contrast and interference levels \cite{jensen1999linear}. 
% Thus, we aim to create a beam pattern with a narrow main lobe and low side lobes, which are two opposing objectives. 

Denote the transducer wavelength by $\lambda=2\pi c/\omega_0$ and the array's aperture size by $L=2(N-1)d$. The angle $\theta_1$ of the first zero in the beam pattern is given by 
\begin{equation}
\sin\theta_1=\frac{\lambda}{L}.
\end{equation} 
Hence, a large array or a high center frequency, yields a narrow main lobe. In contrast, the magnitude of the side lobes is controlled by the weights $\{w_n\}$, known as the aperture function.
The side lobes can be reduced by choosing an aperture function that is smooth like a
Hanning window or a Gaussian shape. This, however, broadens the main lobe width, decreasing system resolution.

Before concluding our discussion on DAS beamforming, we note that in practice ultrasound systems perform beamforming in the digital domain: analog signals are amplified and sampled by an
analog to digital converter \cite{eldar2015sampling}, preceded by an anti-aliasing filter. Sampling rate reduction is discussed in \cite{chernyakova2014fourier}.
In addition, assumptions A\ref{ass:narrowband} and A\ref{ass:farfield} do not typically hold in ultrasound imaging. The signal $f(t)$ is wide-band and imaging is performed in the near-field, leading to time delays that depend non-linearly on both $r$ and $\theta$ as
\begin{equation}
% \tau_n=\frac{1}{2}\left(t+\sqrt{t^2-4\frac{|n|d}{c}t\sin\theta+4\Big(\frac{nd}{c}\Big)^2}\right),
\tau_n = \frac{r+\sqrt{r^2-2ndr\sin\theta+(nd)^2}}{c}.
\end{equation}
% where we substitute $r=ct$.
However, the approach taken here is convenient in introducing the major concepts such as in lobe and side lobes which affect the image quality \cite{jensen1999linear} and is standard in the literature.

\begin{figure}[h]
 \centering
 \includegraphics[trim={2cm 3cm 2cm 3cm},clip,height = 4cm, width = 0.9\linewidth]{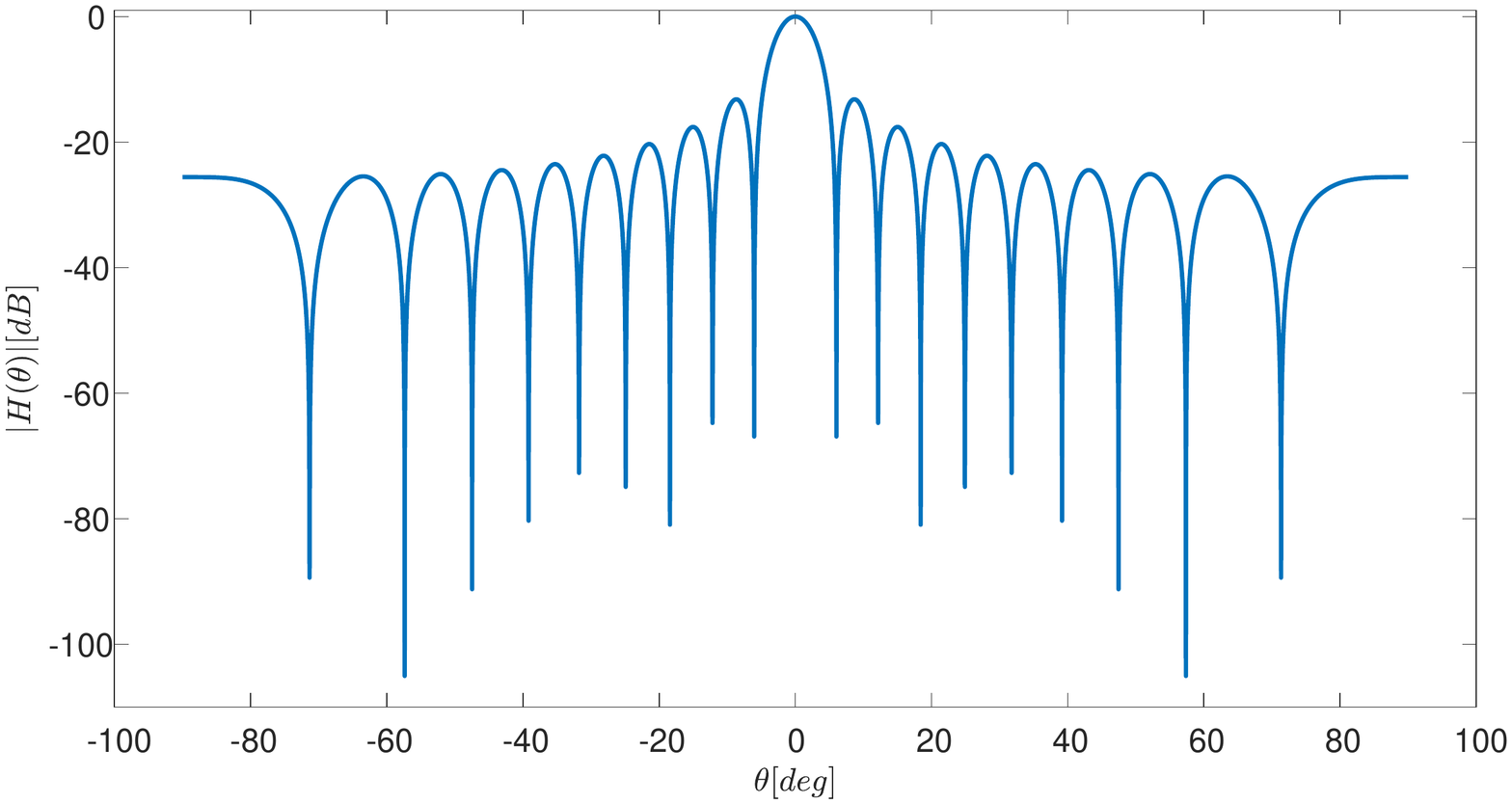}
 \caption{Beam pattern magnitude as a function of angle for $N = 10,\,w_n=1,\lambda=1,d=\frac{1}{2}$.}
  \label{fig:beampattern}
 \end{figure}  

\subsection{Problem Formulation} 

% For ease of analysis, we introduce an auxiliary variable $p$ defined as
% \begin{equation}
% p \triangleq  e^{-2\pi j\frac{f_0\sin\theta}{c}d_x}.
% \label{eq:p}
% \end{equation}
% By substituting (\ref{eq:p}) into (\ref{eq:beampattern}), the beam pattern can be expressed as a polynomial
% \begin{equation}
% H(p)=\sum_{n=-(N-1)}^{N-1} w_np^n,
% \label{eq:bpp}
% \end{equation}
% where every element contributes to the sum a different power of $p$ according to its location and the order of the polynomial determines the aperture size. We shall henceforth refer to (\ref{eq:bpp}) as the beam pattern polynomial and we use this representation for analyzing and designing new beamforming techniques in the following sections.   

The goal of this work is to design arrays with fewer elements than $2N-1$ together with a beamforming method which enables obtaining the beam pattern given by (\ref{eq:das}) or an improved pattern in terms of resolution and image contrast. 
To that end, we first introduce a new beamformer based on a lateral convolution operation and show that it leads to improved resolution by analyzing its beam pattern. Next, we propose two sparse beamforming techniques. The first beamformer uses fewer channels and demonstrates a lateral resolution similar to that of DAS, whereas the second beamformer achieves a twofold improvement in resolution at the expense of a smaller element reduction. Analysis of these approaches shows that the minimal number of elements required to obtain the desired beam patterns is proportional to $\sqrt{N}$.

Throughout the paper, we assume the element pitch $d$ and the transducer center frequency $\omega_0$ are fixed. In addition, we constraint the array configurations so that all element locations satisfy $|x|\leq L/2$. We show that this limitation on the physical array aperture does not prevent us from creating an effective aperture which is larger in size than $L$. 
Note that we assume an odd number of elements $2N-1$ only for clarity of presentation so that the center of the array is well-defined. However, the results presented hold also for an even number of elements.

\section{Convolutional Beamforming and its \\ Beam Pattern}
\label{sec:convbf}
In this section, we present a new non-linear beamformer called COnvolutional Beamforming Algorithm (COBA). The proposed beamformer is based on a convolution operation
% is similar to Filtered Delay Multiply and Sum (F-DMAS) method \cite{lim2008confocal,matrone2015delay}, however, in contrast to the latter, it 
% including the self-products that prevents energy loss \cite{su2018efficient}
and can be implemented efficiently using the FFT. We then introduce the concept of sum co-array \cite{hoctor1990unifying,cohen2018optimized} to analyze the beam pattern of COBA, showing it outperforms DAS in terms of resolution and image contrast.
% Then, we devise two sparse variations of COBA which require a low number of elements that is on the order of $\sqrt{N}$.     

\subsection{Convolutional Beamforming}
\label{subsec:convbf}
Consider the delayed signals $y_n(t)$ given by (\ref{eq:nthoutput}) where $g_n(t)=w_n\delta(t)$ as in DAS.
For simplicity, we assume unity weights $w_n=1$. An extension for arbitrary apodization is given in Section\,\ref{subsec:apod}. 
% Note that the beam pattern polynomial which corresponds to DAS processing is given by (\ref{eq:bpp}).
Inspired by the work on transmit/receive pair array synthesis \cite{hoctor1990unifying}, we define a new beamformed signal as 
\begin{equation}
\bar{y}(t)= \sum_{n=-(N-1)}^{N-1}\sum_{m=-(N-1)}^{N-1} u_n(t)u_m(t),
\label{eq:products}
\end{equation}
where
\begin{equation}
u_n(t)=\exp\{j\phase{y_n(t)}\}\sqrt{|y_n(t)|},\quad -(N-1)\leq n\leq N-1, 
\label{eq:un}
\end{equation}
with $\phase{y_n(t)}$ and $|y_n(t)|$ being the phase and modulus of $y_n(t)$ respectively. 
The operation in (\ref{eq:un}) ensures that the amplitude of each product in (\ref{eq:products}) is on the same order of that of the RF signals $y_n(t)$. This in turn means that the dynamic range of the resultant image will be similar to that obtained by DAS.

Computing (\ref{eq:products}) requires all possible signal pair combinations, i.e., $\binom{2N-1}{2}$ multiplications. Thus, conventionally the computation load for each pixel is $\mathcal{O}(N^2)$, which may lead to slow runtime. This complexity can be substantially reduced by noticing that the beamformed output (\ref{eq:products}) is equivalent to
\begin{equation}
\bar{y}(t)= \sum_{n=-2(N-1)}^{2(N-1)} s_n(t),
\label{eq:bconv}
\end{equation} 
where 
\begin{equation}
s_n = \sum_{(i,j:\,i+j=n)} u_i(t)u_j(t),\quad n=-2(N-1),...,2(N-1).
\label{eq:sn}
\end{equation}
Defining ${\bf s}(t)$ and ${\bf u}(t)$ as the length $2N-1$ vectors whose entries are $s_n(t)$ and $u_n(t)$ receptively, we have that
\begin{equation}
{\bf s}(t) = {\bf u}(t)\underset{s}{\ast}{\bf u}(t),
% =\mathcal{F}^{-1}\big\{{\bf Y}(f)\odot{\bf Y}(f)\big\}
\label{eq:sconv}
\end{equation}  
where $\underset{s}{\ast}$ denotes a discrete linear convolution in the lateral direction. Thus, the vector $\bf s$ can be computed using an FFT by zero-padding $\bf u$ to length $2N-1$, compute the Fourier transform of the result, square each entry and then perform the inverse Fourier transform to get $\bf s$. Thus, the beamformed signal $\bar{y}(t)$ is obtained with low complexity of $\mathcal{O}(N\log N)$ operations.

The temporal products comprising the signal $\bar{y}(t)$ translate to a convolution in the frequency domain with respect to the axial direction, leading to direct current (DC) and harmonic components in the spectrum of $\bar{y}(t)$ \cite{matrone2015delay,park2016delay}. Thus, an additional processing step is required to remove the baseband. The final output of our convolutional beamformer is given by 
\begin{equation}
y_\text{\tiny COBA}(t) = h_\text{BP}(t)\underset{t}{\ast}\bar{y}(t), 
\label{eq:coba}
\end{equation} 
where $h_\text{BP}(t)$ is a band-pass (BP) filter centered at the harmonic frequency $2\omega_0$. 
% The time index is added in (\ref{eq:coba}) to emphasize that the BP filter is applied in the axial direction.
A summary of convolutional beamforming is presented in Algorithm \ref{alg:coba} where the choice of weights is explained in Section \ref{subsec:apod}.

\begin{algorithm}
\caption{\small {\bf CO}nvolutional {\bf B}eamforming {\bf A}lgorithm (COBA)}
\label{alg:coba}
{\fontsize{9}{15}\selectfont
\begin{algorithmic} 
\REQUIRE Delayed RF signals $\{y_n(t)\}$, weights $\{w_n\}$. 
\STATE {\bf 1:} Compute $u_n(t)=\exp\{j\phase{y_n(t)}\}\sqrt{|y_n(t)|}.$ 
\STATE {\bf 2:} Set weights $\tilde{w}_n=\frac{w_n}{(2N-1)-|n|}.$ 
\STATE {\bf 3:} Calculate ${\bf s}(t)={\bf u}(t)\underset{s}{\ast}{\bf u}(t)$ using FFT.  
\STATE {\bf 4:} Evaluate the weighted sum
\begin{equation*}
\bar{y}(t)= \sum_{n=-2(N-1)}^{2(N-1)} \tilde{w}_ns_n(t) .
\end{equation*}
\STATE {\bf 5:}  Apply a band-pass filter 
\begin{equation*}
y_\text{\tiny COBA}(t) = h_\text{BP}(t)\underset{t}{\ast} \bar{y}(t). 
\end{equation*}
\ENSURE Beamformed signal $y_\text{\tiny COBA}(t)$.  
\end{algorithmic}}
\end{algorithm}

We note that COBA involves computing pair-wise products of the RF signals as in FDMAS. However, in contrast to FDMAS, it consists of all possible products, including the self-products
for $n = m$ and repetitions created by interchanging $n$ and $m$. This allows to avoid the high computational complexity and partial energy loss which FDMAS suffers from \cite{su2018efficient}.
In addition, the works related to FDMAS did not consider element reduction which is the main contribution of this paper, and is described in Section \ref{sec:sparseconvbf}.

\subsection{Beam Pattern Analysis}
We now analyze the beam pattern of the proposed convolutional beamformer to show that it outperforms standard DAS beamforming in terms of lateral resolution and image contrast. To this end, we use the following definitions.
\theoremstyle{definition}
\begin{definition}{Position Set:}
Consider a linear array with $d$ being the minimum spacing of the underlying grid on which sensors are assumed to be located. The \textit{position set} is defined as an integer set $I$ where $n\in I$ if there is a sensor located at $nd$.
\label{def:positionset}
\end{definition}
In the interest of brevity, we refer to a linear array with position set $I$ as a linear array $I$.
\theoremstyle{definition}
\begin{definition}{Sum Co-Array:}
Consider a linear array $I$. Define the set
\begin{equation}
\widetilde{S_I} = \{n+m:\quad n,m\in I\}.
\end{equation} 
Note that $\widetilde{S_I}$ includes repetitions of its elements. We also define the set $S_I$, referred to as the sumset of $I$, which consists of the distinct elements of $\widetilde{S_I}$.
The \textit{sum co-array} of $I$ is defined as the array whose position set is $S_I$.
\label{def:sumarray}
\end{definition}
As an example, the sum co-array of an $M$ element ULA is another ULA with $2M-1$ elements. The number of elements in the sum co-array directly determines the number of non-zeros in the convolutional signal given by (\ref{eq:sconv}). 
\theoremstyle{definition}
\begin{definition}{Intrinsic Apodization:} 
Consider a linear array $I$ and define a binary vector $\mathbbm{1}_I$ whose $n$th entry is 1 if $n\in I$ and zero otherwise.
The \textit{intrinsic apodization} is an integer vector defined as 
\begin{equation}
{\bf a}=\mathbbm{1}_I\ast \mathbbm{1}_I.
\label{eq:intrinsicapodization}
\end{equation}
The intrinsic apodization vector is related to $S_I$ and $\widetilde{S_I}$ in the following way. For every $n\in S_I$ the entry $a_n$ denotes the number of occurrences of $n$ in $\widetilde{S_I}$.
\label{def:intapod}
\end{definition}
% As later shown, the convolutional beamformer has an intrinsic apodization determined by the weight function.

To derive an expression for the beam pattern of the convolutional beamformer we assume the input signal to be $f(t)=e^{j\omega_0t}$ impinging on the array at direction $\theta$, as in Section \ref{sec:model}. Consequently, we obtain 
\begin{equation}
u_n(t)=e^{j\omega_0t}e^{-j\omega_0\tau_n},
\label{eq:bpun}
\end{equation} 
where $\tau_n$ is given by (\ref{eq:newfn}). Substituting (\ref{eq:bpun}) into (\ref{eq:products}) we have
\begin{align}
\begin{split}
\bar{y}(t)&= \sum_{n=-(N-1)}^{N-1} \sum_{m=-(N-1)}^{N-1} 
e^{j2\omega_0t}e^{-j\omega_0(\tau_n+\tau_m)}\\
&=e^{j2\omega_0t} \sum_{n,m=-(N-1)}^{N-1} e^{-j\omega_0(\tau_n+\tau_m)}.
\end{split}
\end{align} 
Following band-pass filtering we get
\begin{equation}
y_\text{\tiny COBA}(t)=\left(h_\text{BP}(t)\underset{t}{\ast} e^{j2\omega_0t}\right) 
\sum_{n,m=-(N-1)}^{N-1} e^{-j\omega_0(\tau_n+\tau_m)}.
\end{equation}
In the Fourier domain
\begin{equation}
Y_\text{\tiny COBA}(\omega)=\delta(\omega-2\omega_0)H_\text{BP}(2\omega_0) \sum_{n,m=-(N-1)}^{(N-1)} e^{-j\omega_0(\tau_n+\tau_m)},
\end{equation}
where $H_\text{BP}(\omega)$ is the Fourier transform of the band-pass filter $h_\text{BP}(t)$. Assuming that $H_\text{BP}(2\omega_0)=1$, the beam pattern generated by COBA is
\begin{align}
\begin{split}
H_\text{\tiny COBA}(\theta)&=\sum_{n,m=-(N-1)}^{N-1}e^{-j\omega_0(\tau_n+\tau_m)} \\
&= \sum_{n,m=-(N-1)}^{N-1} e^{-j\omega_0\frac{d\sin\theta}{c}(n+m)},
\label{eq:hpp}
\end{split}
\end{align}
where the last equation is obtained by substituting the explicit expression for $\tau_n$.
% The sum co-array appears naturally in the terms of $n+m$ (\ref{eq:hpp}).

The sum in (\ref{eq:hpp}) is the product of two polynomials $H_\text{\tiny COBA}(\theta)=H_\text{\tiny DAS}(\theta)H_\text{\tiny DAS}(\theta)$ with $H_\text{\tiny DAS}(\theta)$ given by (\ref{eq:das}) assuming $w_n=1$. In Appendix \ref{app:polyprod} we show that $H_\text{\tiny COBA}(\theta)$ can be written as a single polynomial
\begin{equation}
H_\text{\tiny  COBA}(\theta) = \sum_{n=-2(N-1)}^{2(N-1)} a_ne^{-j\omega_0\frac{d\sin\theta}{c}n},
\label{eq:bpcoba}
\end{equation}
where $\{a_n\}$ are triangle-shaped intrinsic apodization weights given by (\ref{eq:intrinsicapodization}). This apodization is illustrated in Fig. \ref{fig:apod}(b) and is further discussed in Section \ref{subsec:apod}.

Equation (\ref{eq:bpcoba}) can be thought of as the beam pattern of a DAS beamformer operating on the sum co-array. This virtual array is twice the size of the physical one, leading to a resolution. 
that is twice better the standard resolution.
In addition, the apodization of the sum co-array reduces the side lobes, thus, the convolutional beamformer results in improved image contrast, as demonstrated in Fig. \ref{fig:beampatterncoba}.    

\begin{figure}[h]
 \centering
 \includegraphics[trim={2cm 3cm 2cm 3cm},clip,height = 4cm, width = 0.9\linewidth]{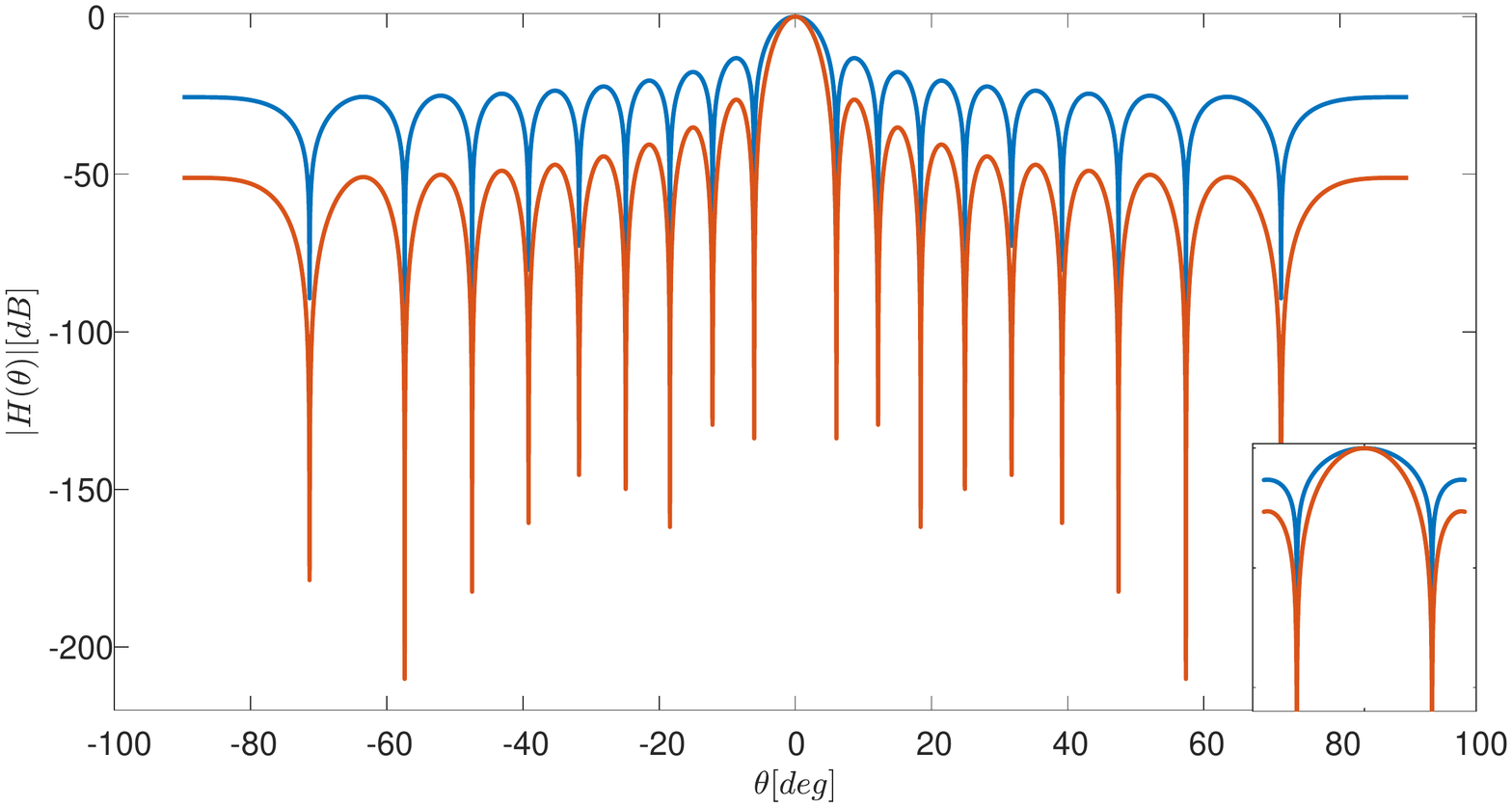}
 \caption{Beam pattern magnitude of DAS (blue) and COBA (red) for $N = 10,\,w_n=1,\lambda=1,d=\frac{1}{2}$. Right bottom corner - zoom in on the main lobes.}
  \label{fig:beampatterncoba}
 \end{figure}

\section{Sparse Convolutional Beamforming}
\label{sec:sparseconvbf}

So far, we presented a convolutional beamformer which leads to better resolution and contrast with respect to DAS. An analysis of its beam pattern showed that its performance depends on the sum co-array, rather than the physical array. 
In this section, we exploit this property to derive two families of beamformers that rely on a reduced number of elements, without affecting the beam pattern.   

\subsection{Sparse Arrays}
Given a ULA of $2N-1$ elements with position set $I=\{-(N-1),...,(N-1)\}$, we aim to remove some of its elements to create a thinned array. The challenge is to design such an array without degrading image quality. To this end, we define the following.
\theoremstyle{definition}
\begin{definition}{Sparse Array:}
Consider a ULA with position set $I$.
A \textit{sparse array} with respect to $I$ is a thinned array, created by removing part of the elements, which satisfies
\begin{equation}
J\subset I \subseteq S_J,
\label{eq:J<I<SJ}
\end{equation}  
where $J$ and $S_J$ are integer sets indicating the elements positions of the thinned array and of its sum co-array respectively. 
\label{def:sa}
\end{definition}  

A sparse array $J$ according to the definition above must be a strict sub-array of $I$, i.e., the number of elements in $J$ is strictly smaller than that of $I$. In addition, performing convolutional beamforming using $J$ is equivalent to applying DAS beamforming on the sum co-array $S_J$ which by Definition \ref{def:sa} has an aperture at least as large as the original ULA $I$. Thus, it results in a beam pattern which is equal or better in resolution than the beam pattern generated by a DAS beamformer applied to $I$.

% Given a ULA, finding a sparse array with fewest elements is a hard combinatorial problem. Here, we provide a simple closed form sparse array design which leads to significant element reduction. In this subsection, our goal is to design sparse arrays which attain a resolution similar to that of a standard DAS. In following subsection, we suggest an array design for which $S_J=S_I$, leading to twice the standard resolution.    

\subsection{Sparse Beamforming}
Here, we provide a simple closed form sparse array design which leads to a large element reduction.

Assume $N$ is not prime, so that it can be factored as $N=AB$ where $A,B\in\mathbb{N}^+$. Given such a decomposition we define the following array:
\begin{align}
\begin{split}
&U_A=\{-(A-1),...\,,0,...\,,A-1\}, \\ 
&U_B=\{nA:\quad n=-(B-1),...\,,0,...\,,B-1\}.
\label{eq:sumnested}
\end{split}
\end{align}
An illustration of this array for $N=9, A=3$ and $B=3$ is seen in Fig. \ref{fig:arraydesign}.

Let  $U_A+U_B=\{n+m:\, n\in U_A,\,m\in U_B\}$. Then
\begin{align}
\begin{split}
U_A+U_B&=\{-(AB-1),...\,,0,...\,,AB-1\} \\
&=\{-(N-1),...\,,0,...\,,N-1\} \\
&=I.
\end{split}
\end{align}
Thus, denoting by $U\subset I$ the array geometry defined as
\begin{equation}
U=U_A\cup U_B,
\label{eq:scobaset}
\end{equation}
it holds that 
\begin{equation}
I\subset S_U
\label{eq:ISU}
\end{equation}
where $S_U$ is the sumset of $U$. Thus, the family of sets (\ref{eq:scobaset}) satisfy (\ref{eq:J<I<SJ}) where the number of elements in each set is $2A+2B-3$. As an example, for $N=9$ and $A=3,B=3$ the set $U$ has only $9$ elements out of $17$ in the full array, as shown in Fig.\,\,\ref{fig:arraydesign}. We note that the proposed sparse arrays are similar to nested arrays \cite{pal2010nested} used in the array processing literature. However, while nested arrays are related to the difference co-array, the sets (\ref{eq:sumnested}) are synthesized from the sum co-array perspective \cite{cohen2018optimized} and have a smaller physical aperture.

Based on $U$, we propose a Sparse Convolutional Beamforming Algorithm (SCOBA) which computes the following signal
\begin{equation}
\bar{y}_\text{\tiny SCOBA}(t)= \sum_{n\in U}\sum_{m\in U} u_n(t)u_m(t),
\label{eq:scobaproducts}
\end{equation}
where $u_n(t)$ is defined in (\ref{eq:un}). Namely, we perform COBA only on the outputs of the elements in $U$. As before, (\ref{eq:scobaproducts}) can be written using the sum co-array $S_U$
\begin{equation}
\bar{y}_\text{\tiny SCOBA}(t)= \sum_{n\in S_U} s_n(t),
\label{eq:scobaconv}
\end{equation}
where
\begin{equation}
s_n(t)=\sum_{(i,j\in U: i+j=n)} u_i(t)u_j(t).
\label{eq:socbasn}
\end{equation}
The final output of SCOBA is given by
\begin{equation}
y_\text{\tiny SCOBA}(t) = h_\text{BP}(t)\ast\bar{y}_\text{\tiny SCOBA}(t). 
\label{eq:scoba}
\end{equation}
% where
% \begin{equation*}
% \tilde{b}= \sum_{n\in S_U} \hat{s}_n,\quad \tilde{s}_n=\sum_{\substack{i+j=n \\ i,j\in U}} y_iy_j.
% \label{eq:sscoba}
% % \tilde{b}= \sum_{n\in U} \sum_{m\in U} y_ny_m.
% \end{equation*}
Computing (\ref{eq:scobaconv}) can be performed using appropriate zero-padding and FFT in $\mathcal{O}(N\log N)$ operations or directly by pair-wise products with complexity $\mathcal{O}\left((A+B)^2\right)$ which may be lower. The proposed beamformer is summarized in Algorithm\,\,\ref{alg:scoba}.

\begin{algorithm}
\caption{\small {\bf S}parse {\bf COBA} (SCOBA)}
\label{alg:scoba}
{\fontsize{9}{15}\selectfont
\begin{algorithmic} 
\REQUIRE Delayed RF signals $\{y_n\}$, weights $\{w_n\}$, parameters $A,B$. 
\STATE {\bf 1:} Construct the set $U$ using (\ref{eq:scobaset}) and its sumset $S_U$. 
\STATE {\bf 2:} Compute $u_n(t)=\exp\{j\phase{y_n(t)}\}\sqrt{|y_n(t)|},\,  n\in U.$ 
\STATE {\bf 3:} Calculate ${\bf a}=\mathbbm{1}_U\ast \mathbbm{1}_U$.
 % using Definition \ref{def:intapod}. 
\STATE {\bf 4:} Set weights $\tilde{w}_n=\frac{w_n}{a_n},\, n\in S_U.$ 
\STATE {\bf 5:} For all $n\in S_U$ compute $s_n(t)$ using (\ref{eq:socbasn}).
\STATE {\bf 6:} Evaluate the weighted sum
\begin{equation*}
\bar{y}(t)= \sum_{n\in S_U} \tilde{w}_ns_n(t).
\end{equation*}
\STATE {\bf 7:}  Apply a band-pass filter 
\begin{equation*}
y_\text{\tiny SCOBA}(t) = h_\text{BP}(t)\underset{t}{\ast} \bar{y}(t). 
\end{equation*}
\ENSURE Beamformed signal $y_\text{\tiny SCOBA}(t)$.  
\end{algorithmic}}
\end{algorithm}

To analyze the beam pattern produced by SCOBA, we follow the same steps presented in Section \ref{sec:convbf}. This leads to
\begin{align}
\begin{split}
H_\text{\tiny SCOBA}(\theta)&= \sum_{n,m\in U} e^{-j\omega_0\frac{d\sin\theta}{c}(n+m)} \\
&= \sum_{n\in S_U} u_n e^{-j\omega_0\frac{d\sin\theta}{c}n},
\label{eq:scobabp}
\end{split}
\end{align}
where ${\bf u}=\mathbbm{1}_U\ast \mathbbm{1}_U$ is the intrinsic apodization of SCOBA. Notice that (\ref{eq:scobabp}) can be rewritten as
 \begin{equation}
 H_\text{\tiny SCOBA}(\theta)= \sum_{n\in I} u_n e^{-j\omega_0\frac{d\sin\theta}{c}n} +
 \sum_{m\in S_U/I} u_m e^{-j\omega_0\frac{d\sin\theta}{c}m},
 \label{eq:scobabp2}
 \end{equation}
where $S_U/I=\{m\in S_U:\,m\notin I\}$.
The first sum in the equation above ensures the resolution of SCOBA is at least as good as the resolution of a DAS beamformer applied on the full array $I$. The second sum, on the right hand side of (\ref{eq:scobabp2}), provides additional degrees of freedom that may be used to improve the resolution. The image contrast depends on the apodization $\{u_n\}$ which can be adjusted as we describe later in Section \ref{subsec:apod}. A demonstration of the beam pattern of SCOBA is presented in Fig. \ref{fig:beampatternsparse}.

The number of elements required for SCOBA is $2A+2B-3$, thus, it leads to a family of beamformers in which each beamformer demonstrates a different level of element reduction, controlled by the parameters $A,B$. While a large number of elements may be favorable in the presence of noise, it also increases the mutual coupling \cite{merrill2001introduction,liu2016superconf,liu2016super,liu2016high}, which is the electromagnetic interaction between adjacent sensors that has an adverse effect on obtaining a desired beam pattern. In Section \ref{subsec:opt} we discuss how to minimize the number of sensors using this approach.

\begin{figure*}[h]
 \centering
 \includegraphics[trim={3.5cm 4cm 2cm 4cm},clip,height = 9cm, width = 1\linewidth]{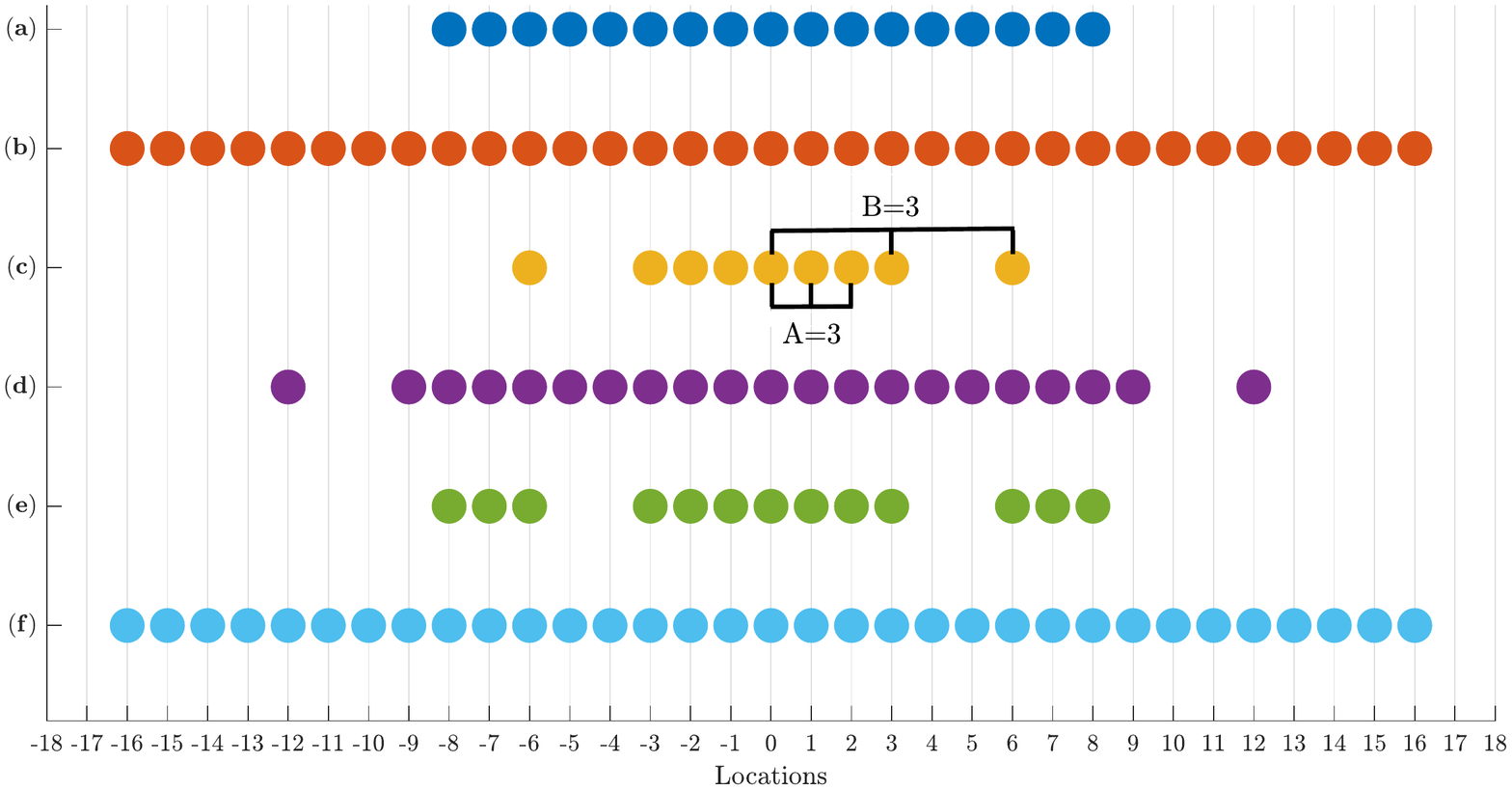}
 \caption{Element positions of (a) ULA $I=[-8,8]$, (b) sum co-array $S_I=[-16,16]$, (c) sparse array $U$ given by (\ref{eq:scobaset}), (d) sum co-array $S_U$, (e) sparse array $V$ defined by (\ref{eq:scobarset}), (f) sum co-array $S_V$. In this example, the element spacing is $d=1$, $N=9$ and $A=3, B=3$.}
  \label{fig:arraydesign}
 \end{figure*}

 \begin{figure}[h]
 \centering
 \includegraphics[trim={2cm 3cm 2cm 3cm},clip,height = 4cm, width = 0.9\linewidth]{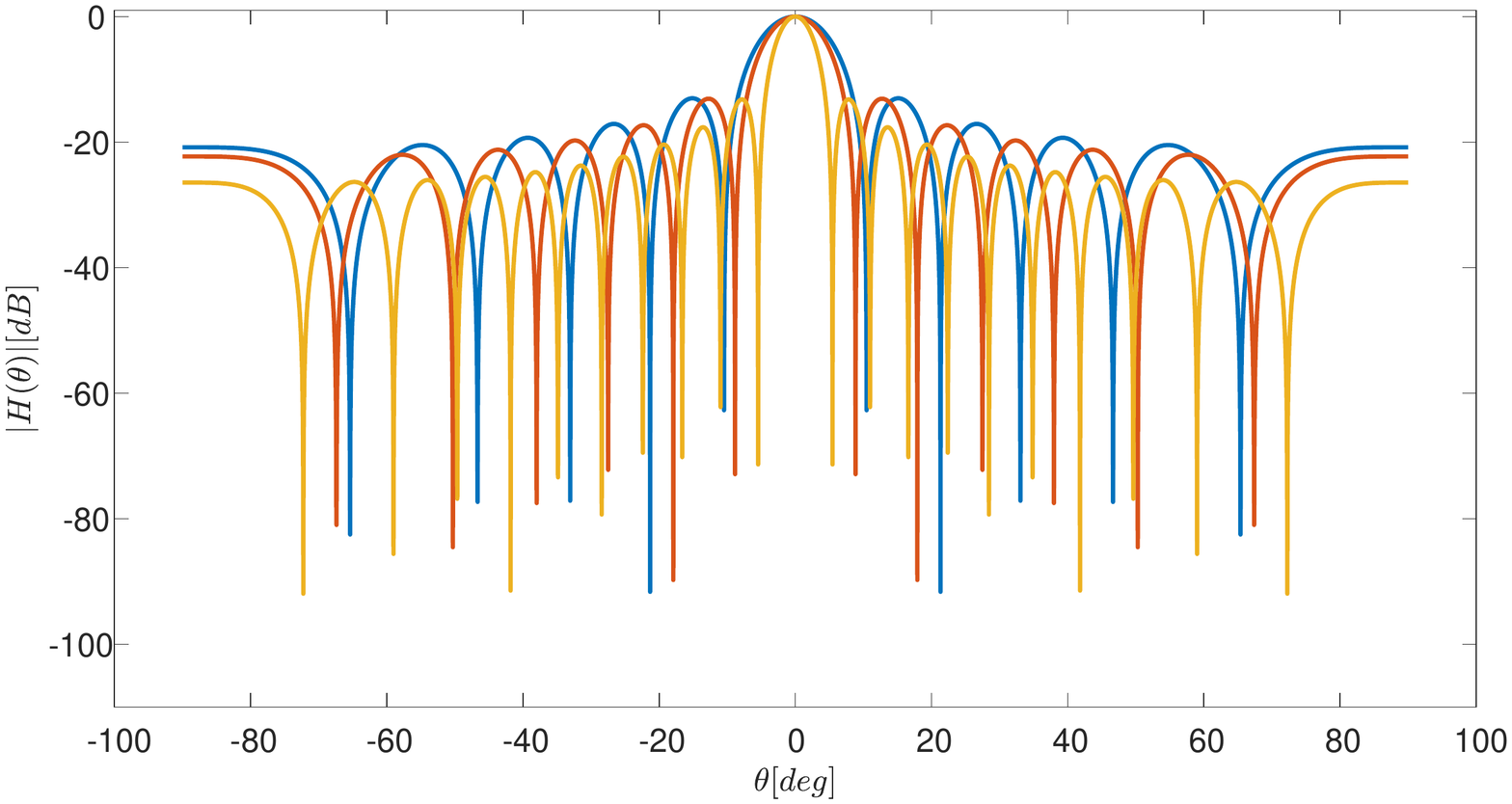}
 \caption{Beam pattern of DAS (blue), SCOBA (red) and SCOBAR (yellow) for $N = 6$ and $A=3, B=2,\lambda=1,d=\frac{1}{2}$.}
  \label{fig:beampatternsparse}
 \end{figure}  

\subsection{Sparse Beamforming with Super Resolution}
Previously we presented COBA which achieves double the standard resolution. Following that, we introduced a sparse array design to create a beamformer which uses fewer elements and yields a resolution that is comparable to the standard one. Now, we propose a family of sparse beamformers with enhanced resolution that is equivalent to that of COBA, thereby combining the best of both worlds. We refer to this technique as Sparse Convolutional Beamforming Algorithm with super-Resolution (SCOBAR).

We extend the array configuration used in SCOBA by constructing an additional array as
\begin{equation}
U_C =\{n:\quad |n|=N-A,...,N-1\}.
% =\{A(B-1),...,N-1\}
\end{equation}
Then, we define a sparse array geometry given by
\begin{equation}
V=U_A\cup U_B\cup U_C.
\label{eq:scobarset}
\end{equation}
As shown in Fig. \ref{fig:arraydesign}, we obtain $V$ by adding to $U$ two small ULAs of size $A-1$ at its edges. It can be verified that
\begin{equation}
V\subset I\subset S_V=S_I,
\label{eq:SVSI}
\end{equation} 
i.e. the sum co-array of $V$ is equal to the sum co-array of the full array $I$.
SCOBAR uses the array sensors given by $V$ to compute the signal
\begin{align}
\begin{split}
\bar{y}_\text{\tiny SCOBAR}(t)&= \sum_{n\in V}\sum_{m\in V} u_n(t)u_m(t) \\
&=\sum_{n\in S_V} s_n(t)
\label{eq:scobarconv}
\end{split}
\end{align}
where $u_n(t)$ is defined in (\ref{eq:un}) and
\begin{equation}
s_n(t)=\sum_{(i,j\in V: i+j=n)} u_i(t)u_j(t).
\label{eq:socbarsn}
\end{equation}
The final output of SCOBAR is given by
\begin{equation}
y_\text{\tiny SCOBAR}(t) = h_\text{BP}(t)\ast\bar{y}_\text{\tiny SCOBAR}(t). 
\label{eq:scobar}
\end{equation}

Similar to SCOBA, (\ref{eq:scobarconv}) can be calculated using FFT in $\mathcal{O}(N\log N)$ operations or directly in $\mathcal{O}\left((A+B)^2\right)$. A summary of SOCBAR is provided in Algorithm \ref{alg:scobar}.

\begin{algorithm}
\caption{\small {\bf SCOBA} with Super-{\bf R}esolution (SCOBAR)}
\label{alg:scobar}
{\fontsize{9}{15}\selectfont
\begin{algorithmic} 
\REQUIRE Delayed RF signals $\{y_n\}$, weights $\{w_n\}$, parameters $A,B$. 
\STATE {\bf 1:} Construct the set $V$ using (\ref{eq:scobarset}) and its sumset $S_V$. 
\STATE {\bf 2:} Compute $u_n(t)=\exp\{j\phase{y_n(t)}\}\sqrt{|y_n(t)|},\,  n\in V.$ 
\STATE {\bf 3:} Calculate ${\bf a}=\mathbbm{1}_V\ast \mathbbm{1}_V$.
 % using Definition \ref{def:intapod}. 
\STATE {\bf 4:} Set weights $\tilde{w}_n=\frac{w_n}{a_n},\, n\in S_V.$ 
\STATE {\bf 5:} For all $n\in S_V$ compute $s_n(t)$ using (\ref{eq:socbarsn}).
\STATE {\bf 6:} Evaluate the weighted sum
\begin{equation*}
\bar{y}_\text{\tiny SCOBAR}(t)= \sum_{n\in S_V} \tilde{w}_ns_n(t).
\end{equation*}
\STATE {\bf 7:}  Apply band-pass filter 
\begin{equation*}
y_\text{SCOBAR}(t) = h_\text{BP}(t)\underset{t}{\ast} \bar{y}_\text{\tiny SCOBAR}(t). 
\end{equation*}
\ENSURE Beamformed signal $y_\text{SCOBAR}(t)$.  
\end{algorithmic}}
\end{algorithm}

Following similar arguments as for COBA and SCOBA, the beam pattern of SCOBAR is 
\begin{align}
\begin{split}
H_\text{\tiny SCOBAR}(\theta)&= \sum_{n,m\in V} e^{-j\omega_0\frac{d\sin\theta}{c}(n+m)} \\
&= \sum_{n\in S_V} v_n e^{-j\omega_0\frac{d\sin\theta}{c}n} \\
&= \sum_{n=-2(N-1)}^{2(N-1)} v_n e^{-j\omega_0\frac{d\sin\theta}{c}n}
\label{eq:scobarbp}
\end{split}
\end{align}
where ${\bf v}=\mathbbm{1}_V\ast \mathbbm{1}_V$ and the last equality is due the fact that $S_V=S_I=\{-2(N-1),...,2(N-1)\}$. The latter implies that the lateral resolution of SCOBAR is similar to that of COBA, twofold better than the standard one, as shown in Fig. \ref{fig:beampatternsparse}. The weights $\{v_n\}$ can be modified to control the image contrast as described in the next subsection. 

The number of elements required by SCOBAR is $4A+2B-5$, thus, the improved resolution is at the expense of a smaller element reduction in comparison with SCOBA. Optimization of the parameters $A$ and $B$ is presented in Section \ref{subsec:opt}.

% \begin{figure*}[h]
%  \centering
%  \includegraphics[trim={3cm 0cm 3cm 0cm},clip, width = 1\linewidth]{Figures/SparseDesign2.eps}
%  \caption{Element locations of (a) full ULA $I=[-8,8]$, (b) sum co-array of $I$, (c) sparse array $V$ defined in (\ref{eq:scobarset}), (d) sum co-array of $V$. In these examples, $d_x=1$ is the element spacing and $A=B=3$.}
%   \label{fig:arraydesign2}
%  \end{figure*}

\subsection{Apodization}
\label{subsec:apod}
As stated in Section \ref{sec:model}, the image contrast is governed by the peak side lobe level. Amplitude apodization \cite{usbook} is an important tool used for suppressing side lobes, leading to improved contrast. Typical apodization functions include Hanning, Hamming, or Gaussian amplitude weighting of the elements which lowers the side lobes at the expense of widening the main lobe width, i.e., worsening the lateral resolution. Therefore, there is a tradeoff between lateral resolution and contrast and a judicious choice of the apodization must be made based on the clinical application.

In the previous section, we introduced the concept of intrinsic apodization which arises in the proposed beamformers. This concept is extended to a standard DAS beamformer by assuming the array has intrinsic weights of unity. Figure\,\ref{fig:apod} shows the intrinsic apodization of the different beamformers. For a DAS beamformer the intrinsic apodization function is constant over the elements and equal to 1, whereas, for COBA we get a triangle-shaped apodization which suppresses side lobes. The intrinsic apodization functions of SCOBA and SCOBAR depend on the parameters $A$ and $B$, and should be analyzed to avoid unwanted side lobes. 

To address this issue, we propose a simple adjustment to the apodization function which takes into account the intrinsic apodization. Given a desired apodization function with weights $\{w_n\}$, we define modified apodization function
\begin{equation}
\tilde{w}_n = \frac{w_n}{a_n},
\label{eq:weights}
\end{equation}
where $\{a_n\}$ are the intrinsic weights assumed to be non-zero. Then, we apply these weights on the sum co-array by computing the weighted sum
\begin{equation}
\bar{y}(t)=\sum_n \tilde{w}_n s_n(t),
\end{equation}
prior to band-pass filtering.
This ensures that the resulting beam pattern will have weights equal to $\{w_n\}$ as desired. 

The intrinsic apodization functions of COBA and SCOBAR both have only non-zeros, thus, any apodization can be achieved using (\ref{eq:weights}). In the case of SCOBA, the intrinsic apodization has zeros, leading to discontinuities in the beam pattern which may be considered as a drawback at first glance. However, note that this is expected since the intrinsic apodization of SCOBA is designed to have non-zeros in the range $-[N-1,N-1]$, similar to DAS beamformer. Thus, any apodization function obtained by DAS can be attained by SCOBA. In fact, the intrinsic apodization of SCOBA has more degrees of freedom (non-zeros) than DAS as shown in (\ref{eq:scobabp2}), allowing the use of an extended family of apodization functions.        

\begin{figure*}[h]
 \centering
 \includegraphics[trim={3cm 4cm 3cm 4cm},clip, height = 7cm, width = 0.8\linewidth]{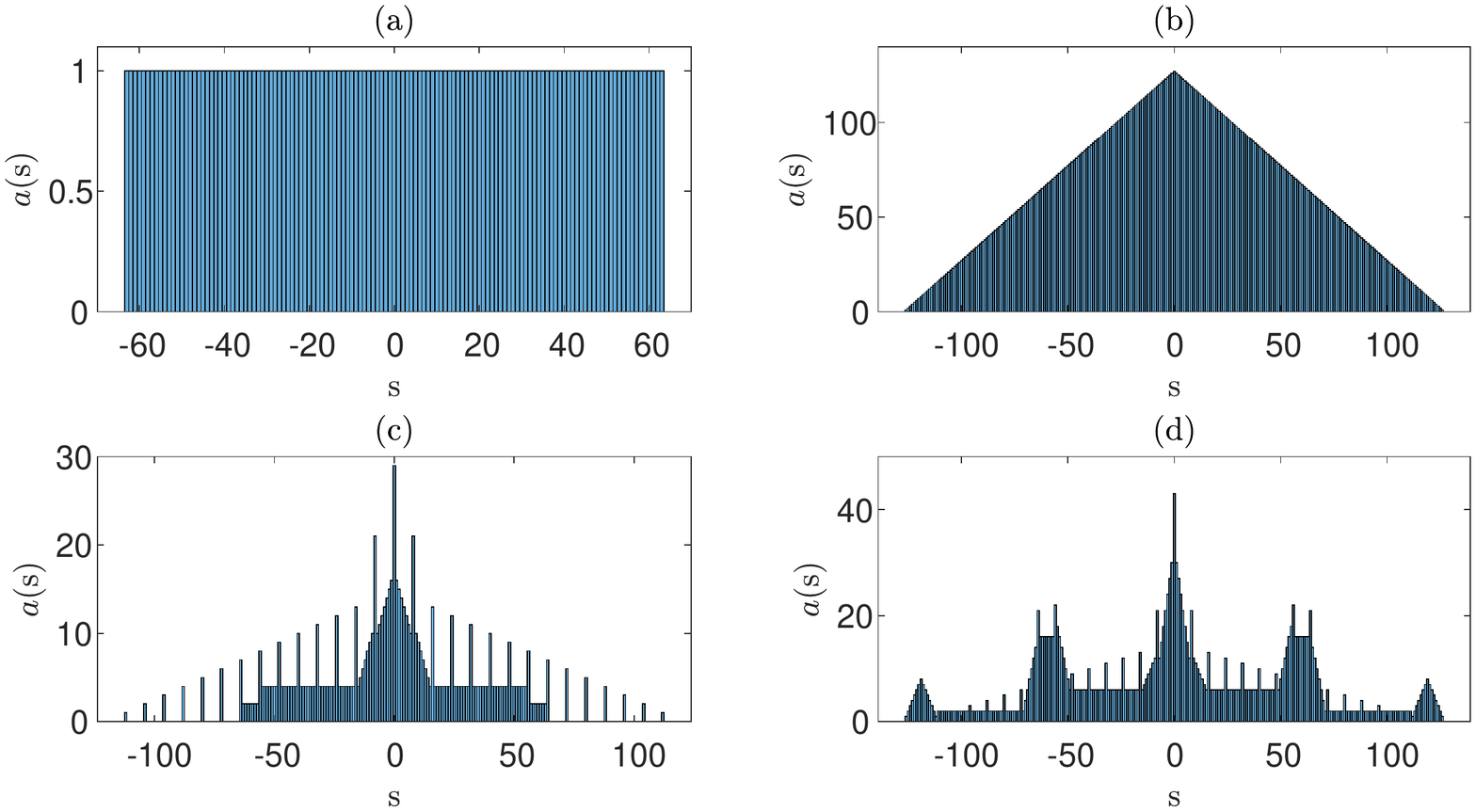}
 \caption{Intrinsic apodization of (a) DAS with 127 elements, (b) COBA with 127 elements, (c) SCOBA with $A=B=8$ and (d) SCOBAR with $A=B=8$.}
  \label{fig:apod}
 \end{figure*}

\subsection{Minimal Number of Elements}
\label{subsec:opt}
As noted before, the number of elements used by SCOBA and SCOBAR is controlled by the parameters $A$ and $B$. We next derive expressions for $A$ and $B$, leading to a minimal number of sensors required by the proposed beamformers. 

For SCOBA, minimizing the number of elements can be cast as the following optimization problem:
\begin{align}
\begin{split}
A^*,B^* =\underset{A,B\in\mathbb{N}}{\arg\min}\quad &2(A+B)-3 \\
\text{subject to}\quad &AB=N.
\label{eq:scobaopt}
\end{split}
\end{align}
When $N$ is a prime number, there are only two feasible solutions which are optimal given by $A=N,\,B=1$ and vice versa; both result in a fully populated array.
Hence, we consider below the case where $N$ is not prime and (\ref{eq:scobaopt}) becomes a combinatorial optimization problem. A closed form solution is presented in the following theorem.

\begin{theorem}
Given an arbitrary $N\in\mathbb{N}^+$, define the sets
\begin{align}
\begin{split}
&D_1 = \big\{m\in\mathbb{N}:\, m|N,\,m\leq\sqrt{N}\big\}, \\
&D_2 = \big\{m\in\mathbb{N}:\, m|N,\,m\geq\sqrt{N}\big\},
\end{split}
\end{align}   
where $m|N$ indicates that $m$ is a divisor of $N$.
The optimal solutions for (\ref{eq:scobaopt}) are
\begin{align}
\begin{split}
&A^*=\max(D_1)\text{ and } B^*=\min(D_2), \\
&A^*=\min(D_2)\text{ and } B^*=\max(D_1).
\end{split}
\end{align}
When $N$ is a perfect square, we have a single optimal solution
\begin{equation}
A^*=B^*=\sqrt{N}.
\end{equation}
\label{theo:scobamin}
\end{theorem}

\begin{proof}
Here we prove the case where $N$ is a perfect square. The complete proof is given in Appendix \ref{app:scobaproof}.

First, we write an equivalent problem to (\ref{eq:scobaopt}) as
\begin{align}
\begin{split}
A^*,B^* =\underset{A,B\in\mathbb{N}}{\arg\min}\quad &A+B \\
\text{subject to}\quad &AB=N. 
\label{eq:scobaopt2}
\end{split}
\end{align}
From the inequality of the arithmetic and geometric means we have that
\begin{equation}
\sqrt{AB}\leq\frac{A+B}{2}.
\label{eq:amgm}
\end{equation}
Thus, the objective function is lower bounded by $2\sqrt{N}$. Equality is obtained in (\ref{eq:amgm}) if and only if $A=B=\sqrt{N}$. Thus, this choice
attains the lower bound and is optimal, which concludes the proof.
\end{proof}

Theorem \ref{theo:scobamin} implies that the minimal number of elements required by SCOBA is proportional to $\sqrt{N}$ and the beamformed signal given by (\ref{eq:scobar}) can be computed with a low complexity of $\mathcal{O}(N)$.

As for SCOBAR, note that when $B=1$, $U_C=U_A$ and $A=N$, leading to the trivial case where the array is full, hence, we assume that $B>1$. In this case, the minimal number of elements required by SCOBAR is given by the solution to
\begin{align}
\begin{split}
A^*,B^* =\underset{A,B\in\mathbb{N},\,B>1}{\arg\min}\quad &2(2A+B)-5 \\
\text{subject to}\quad &AB=N.
\label{eq:scobaropt}
\end{split}
\end{align}
When $N$ is prime, the only feasible and hence optimal solution is $A=1$ and $B=N$ which is trivial. Therefore, we address below the case where $N$ is not prime. A closed form solution is obtained in the next theorem. 

\begin{theorem}
Given an arbitrary $N\in\mathbb{N}^+$, define the sets 
\begin{align}
\begin{split}
&D_{3} = \big\{m\in\mathbb{N}:\, m|2N,\,m\leq\sqrt{2N}\big\}, \\
&D_{4} = \big\{m\in\mathbb{N}:\, m|2N,\,m\geq\sqrt{2N}\big\}. 
\end{split}
\end{align}
Denote by $\mathbb{E}$ the set of even integers. The optimal solutions for (\ref{eq:scobaropt}) are given by the following cases:
\begin{enumerate}[label=(\roman*)]
\item $\max(D_3)\in\mathbb{E},\,\min(D_4)\in\mathbb{E}$
\begin{align}
\begin{split}
&A^*=\max(D_3)/2\text{ and } B^*=\min(D_4), \\
&A^*=\min(D_4)/2\text{ and } B^*=\max(D_3).
\end{split}
\end{align}
\item $\max(D_3)\in\mathbb{E},\,\min(D_4)\notin \mathbb{E}$
\begin{equation}
A^*=\max(D_3)/2\text{ and } B^*=\min(D_4).
\end{equation}
\item $\max(D_3)\notin\mathbb{E},\,\min(D_4)\in \mathbb{E}$
\begin{equation}
A^*=\min(D_4)/2\text{ and } B^*=\max(D_3).
\end{equation}
\end{enumerate}
When $2N$ is a perfect square,  there is a single solution
\begin{equation}
A^*=\frac{\sqrt{2N}}{2},\,B^*=\sqrt{2N}.
\end{equation}
\label{theo:scobarmin}
 \end{theorem}

\begin{proof}
Here we provide the proof only for the special case when $2N$ is a perfect square. The proof for the general case is detailed in Appendix \ref{app:scobarproof}. 

Problem (\ref{eq:scobaropt}) is equivalent to
\begin{align}
\begin{split}
A^*,B^* =\underset{A,B\in\mathbb{N},\,B>1}{\arg\min}\quad &2A+B \\
\text{subject to}\quad &AB=N.
\label{eq:scobaropt2}
\end{split}
\end{align}
Once more, using the inequality of arithmetic and geometric means, we get
\begin{equation}
\sqrt{2AB}\leq\frac{2A+B}{2}.
\end{equation}
Consequently, the objective value is lower bounded by $2\sqrt{2N}$. This bound is attained by choosing $A=\sqrt{\frac{N}{2}}$ and $B=\sqrt{2N}$.
\end{proof}

Theorem \ref{theo:scobarmin} indicates that SCOBAR requires a minimal number of sensors which is on the order of $\sqrt{2N}$. The beamformed signal can be obtained in complexity $\mathcal{O}(N)$, similar to SCOBA. Note, however, that while SCOBAR demonstrates almost twofold improvement in resolution, the increase in the number of elements, compared to SCOBA, is roughly only by a factor of $\sqrt{2}$. 

\subsection{Minimal Physical Aperture}
\label{subsec:apt}
With the purpose of reducing cost and size, one may desire to design a compact probe with a small physical aperture. While the size of the physical aperture using COBA and SCOBAR is fixed and given by $L=(2N-1)d$, for SCOBA it is equal to $\tilde{L}=2A(B-1)d$ where $B>1$, and thus can be minimized using an appropriate choice of $A$ and $B$. This objective can be formulated as follows 
\begin{align}
\begin{split}
A^*,B^* =\underset{A,B\in\mathbb{N},\,B>1}{\arg\min}\quad &A(B-1) \\
\text{subject to}\quad &AB=N.
\label{eq:minapt}
\end{split}
\end{align}
The solution to (\ref{eq:minapt}) is given by the next theorem. 
\begin{theorem}
Consider a non-prime number $N\in\mathbb{N}^+$. Denote by $D$ the set of the non-trivial divisors of $N$, defined as
\begin{equation}
D = \big\{m\in\mathbb{N}:\, m|N,\,1<m<N\}.
\end{equation}
Then, the optimal solution to (\ref{eq:minapt}) is   
\begin{equation}
A^*=\max(D),B^*=\min(D).
\end{equation}
\label{theo:minapt}
 \end{theorem}

 \begin{proof}
Using the fact that $A(B-1)=N-A$, we rewrite (\ref{eq:minapt}) as
\begin{align}
\begin{split}
A^*,B^* =\underset{A,B\in\mathbb{N},\,B>1}{\arg\max}\quad &A \\
\text{subject to}\quad &AB=N.
\label{eq:minapt2}
\end{split}
\end{align}
It is easy to see from (\ref{eq:minapt2}) that the optimal $A$ is the maximal non-trivial divisor of $N$, i.e., $A^*=\max(D)$. Consequently, $B^*=\frac{N}{A^*}=\min(D)$ which concludes the proof. 
\end{proof}

Theorem \ref{theo:minapt} implies that when $N=2M$ with $M\in\mathbb{N}^+$ the optimal choice is $A=M$ and $B=2$, which leads to a ULA with a physical aperture that is twice of that of the original ULA. In other words, performing SCOBA on a given ULA is equivalent to performing COBA on a ULA with half the size. Note, however, that the number of elements in this case is $2M+1=N+1$, which is much larger than the minimal number achieved by Theorem \ref{theo:scobamin}. 

% \section{Summary of the Proposed Beamformers}

\section{Evaluation Results}
\label{sec:results}
We now verify the performance of the proposed beamformers in comparison to DAS. The resolution and contrast are  first evaluated using Field II simulator \cite{jensen1996field,jensen1992calculation} in MATLAB. Following that, we apply the methods on phantom data, scanned using a Verasonics imaging system, and on in vivo cardiac data acquired from a healthy volunteer.

In the following experiments, we do no apply apodization upon reception for DAS and COBA. For fair comparison, we employ weights in SCOBA to create an effective apodization of ones as in DAS. For SCOBAR we apply weights to yield an effective triangle-shaped apodization as in COBA. The full transducer array is used for transmission and element reduction is performed only on the receive end.

% For clarity and practical purposes, we summarize the proposed beamforming techniques COBA, SCOBA and SCOBAR in Algorithm \ref{alg:coba}, Algorithm \ref{alg:scoba} and Algorithm \ref{alg:scobar} respectively. 

\subsection{Simulations}
In both simulations presented here, we used an array consisting of 127 elements with 
an element width of 440 $\mu m$, a height of 6 mm and a kerf of 0.0025 mm. During transmission, the transducer generated a Hanning-windowed 2-cycle sinusoidal pulse with a center frequency of 3.5 MHz and a focal depth of 50 mm.

In COBA, SCOBA and SCOBAR a BP filter was applied using a Hanning window. The window frequency boundaries were empirically determined to well isolate the signal band to be preserved (See Fig. \ref{fig:filter}). The sampling frequency was 100 MHz. For SCOBA and SCOBAR we used $A=B=8$, which leads to the minimal numbers of 29 (23\%) and 43 (34\%) elements respectively, according to Theorem \ref{theo:scobamin} and Theorem\,\ref{theo:scobarmin}.

\begin{figure}[h]
 \centering
 \includegraphics[trim={3cm 3cm 3cm 3cm},clip,height = 4cm, width = 0.8\linewidth]{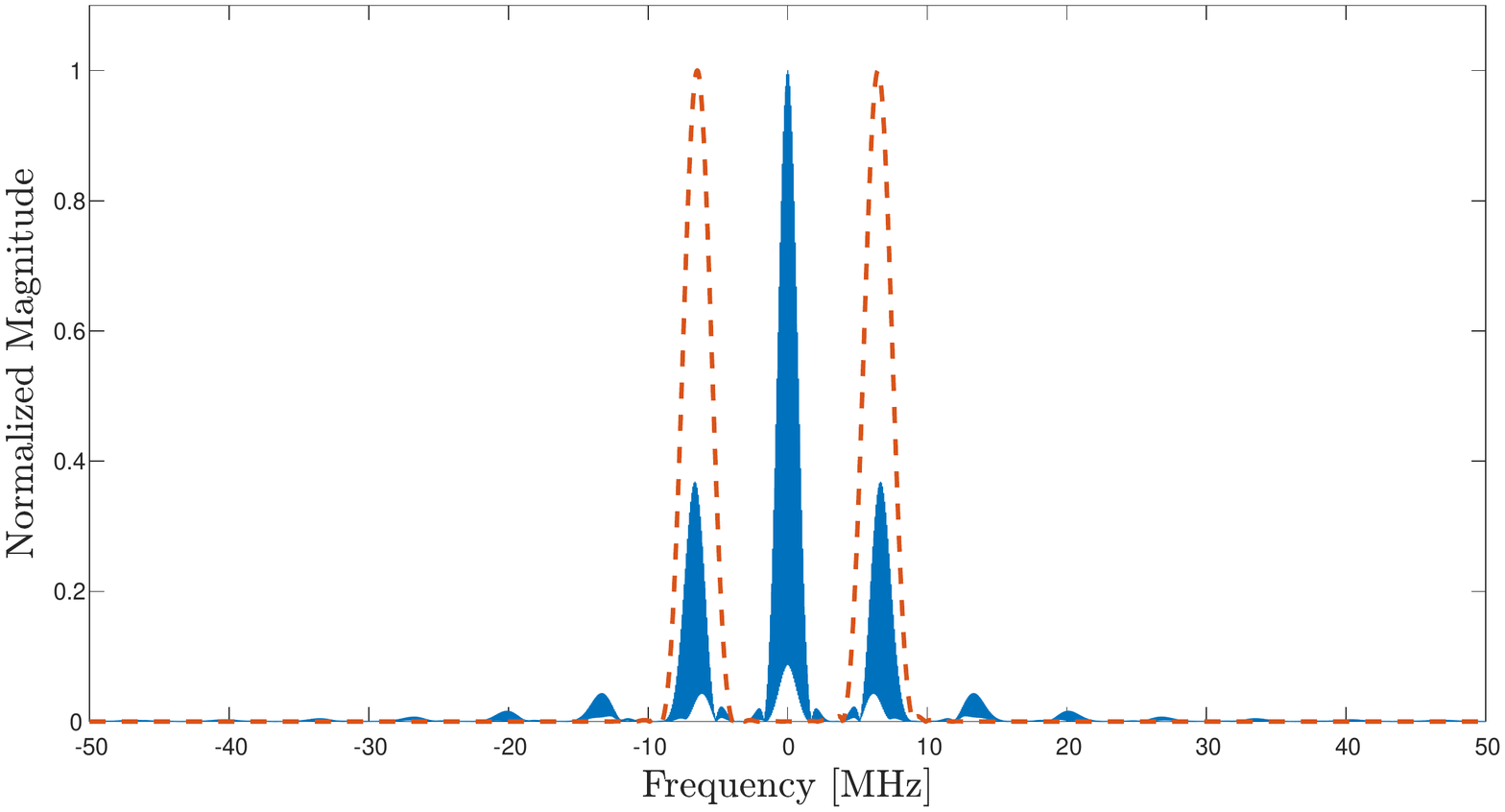}
 \caption{Fourier transforms of the impulse response of the Hanning-based BP filter (dashed line) and of the signal given by ($\ref{eq:bconv}$), which corresponds to the central simulated image line.}
  \label{fig:filter}
 \end{figure}

\subsubsection{Resolution}
We evaluate resolution using a point-reflector simulated phantom with isolated scatterers distributed in an anechoic background. Figure \ref{fig:imres} presents the results of DAS, COBA, SCOBA and SCOBA.
As seen from the images SCOBA achieved a comparable lateral resolution to that of DAS while using fewer elements. COBA outperforms DAS in terms of lateral resolution which is seen clearly in the focal depth and beyond it. SCOBAR obtains similar results to COBA using fewer elements. For a closer look, the center image line and the lateral cross-section of the   
scattering point placed at the transmission focus in 50 mm are shown in Fig.\,\ref{fig:resolution}. One can observe from Fig.\,\ref{fig:resolution}(a) that all four methods have similar axial resolution. In terms of lateral resolution, the performance of SCOBA is the same as DAS, while SCOBAR is better than DAS and COBA outperforms them all. Figure\,\ref{fig:offaxis} shows similar results obtained through a simulation which includes on-axis targets as well as off-axis targets. 

\begin{figure*}[h]
 \centering
 \includegraphics[trim={3cm 1cm 1cm 2cm},clip, height = 9cm, width = 0.9\linewidth]{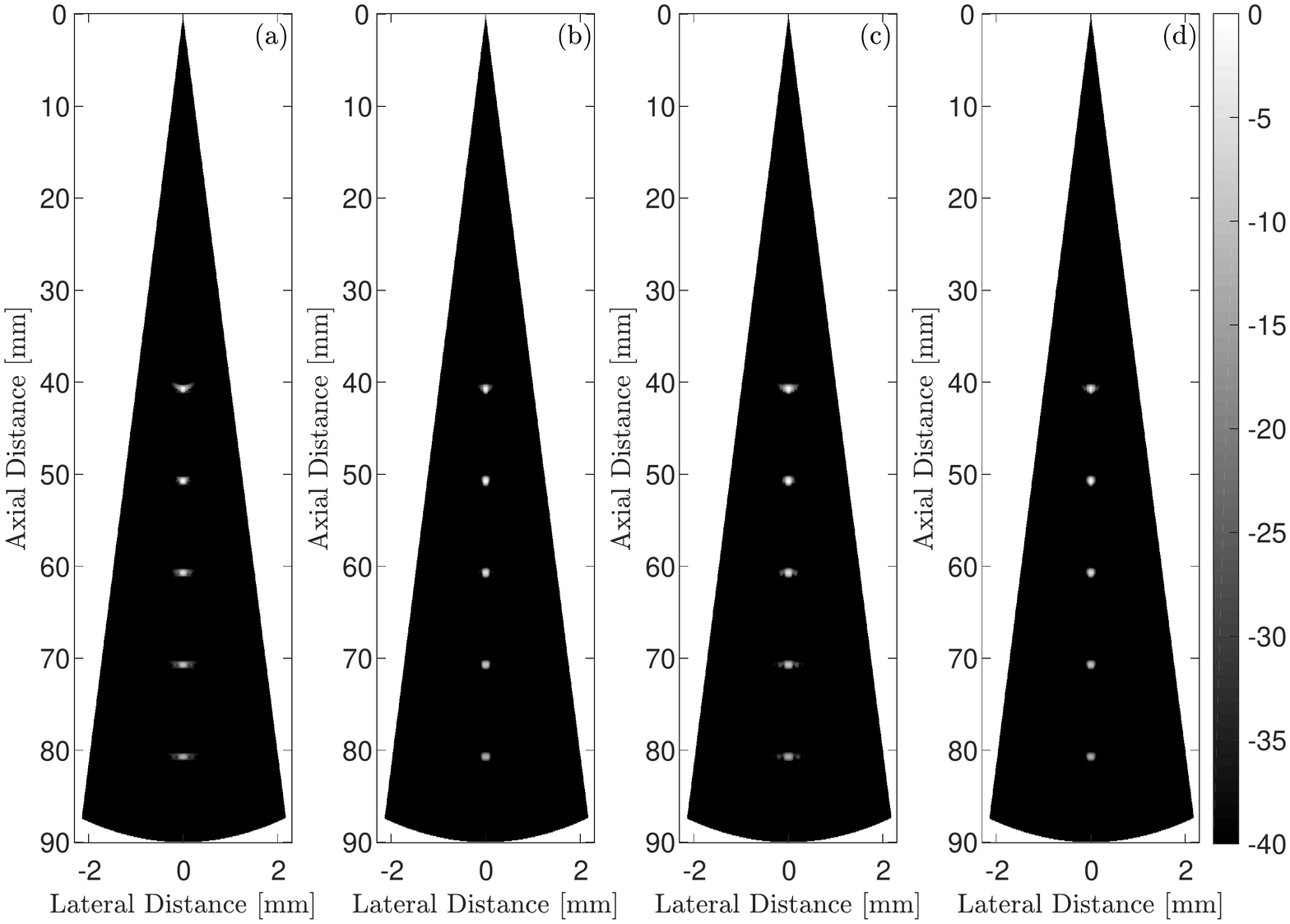}
 \caption{Images of simulated point-reflector phantom obtained by (a) DAS (127), (b) COBA (127), (c) SCOBA (29), (d) SCOBAR (43). 
%  The dynamic range for all the images is 40 dB. 
 Number in brackets refers to the number of elements used.}
  \label{fig:imres}
 \end{figure*}

\begin{figure*}
\centering
\begin{subfigure}{.5\textwidth}
  \centering
  \includegraphics[trim={2cm 3.5cm 1cm 4cm},clip, height = 3.5cm, width = 0.8\linewidth]{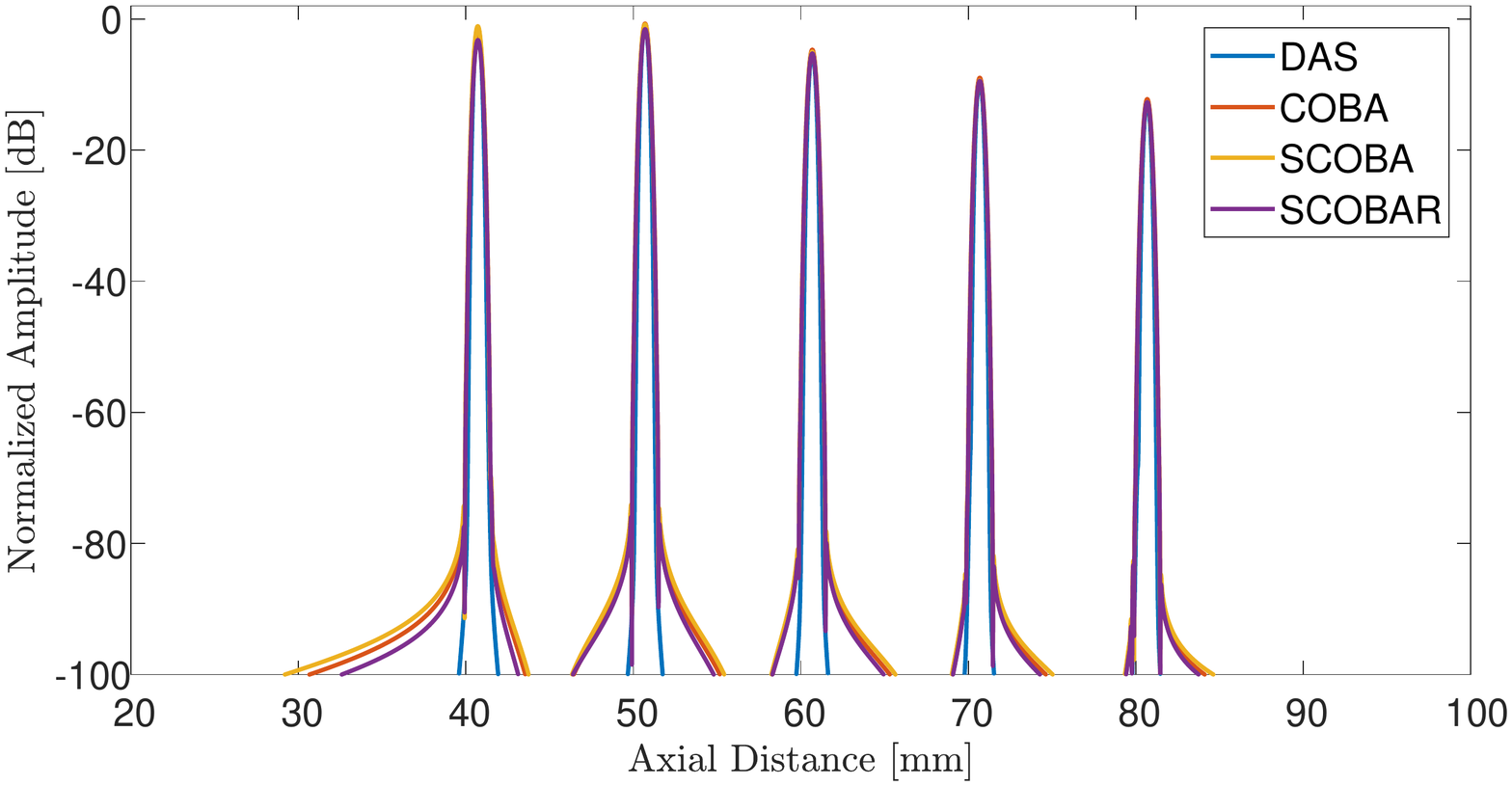}
  \caption{}
  \label{fig:sub1}
\end{subfigure}%
\begin{subfigure}{.5\textwidth}
  \centering
  \includegraphics[trim={2cm 3.5cm 1cm 4cm},clip, height = 4cm, width = 0.9\linewidth]{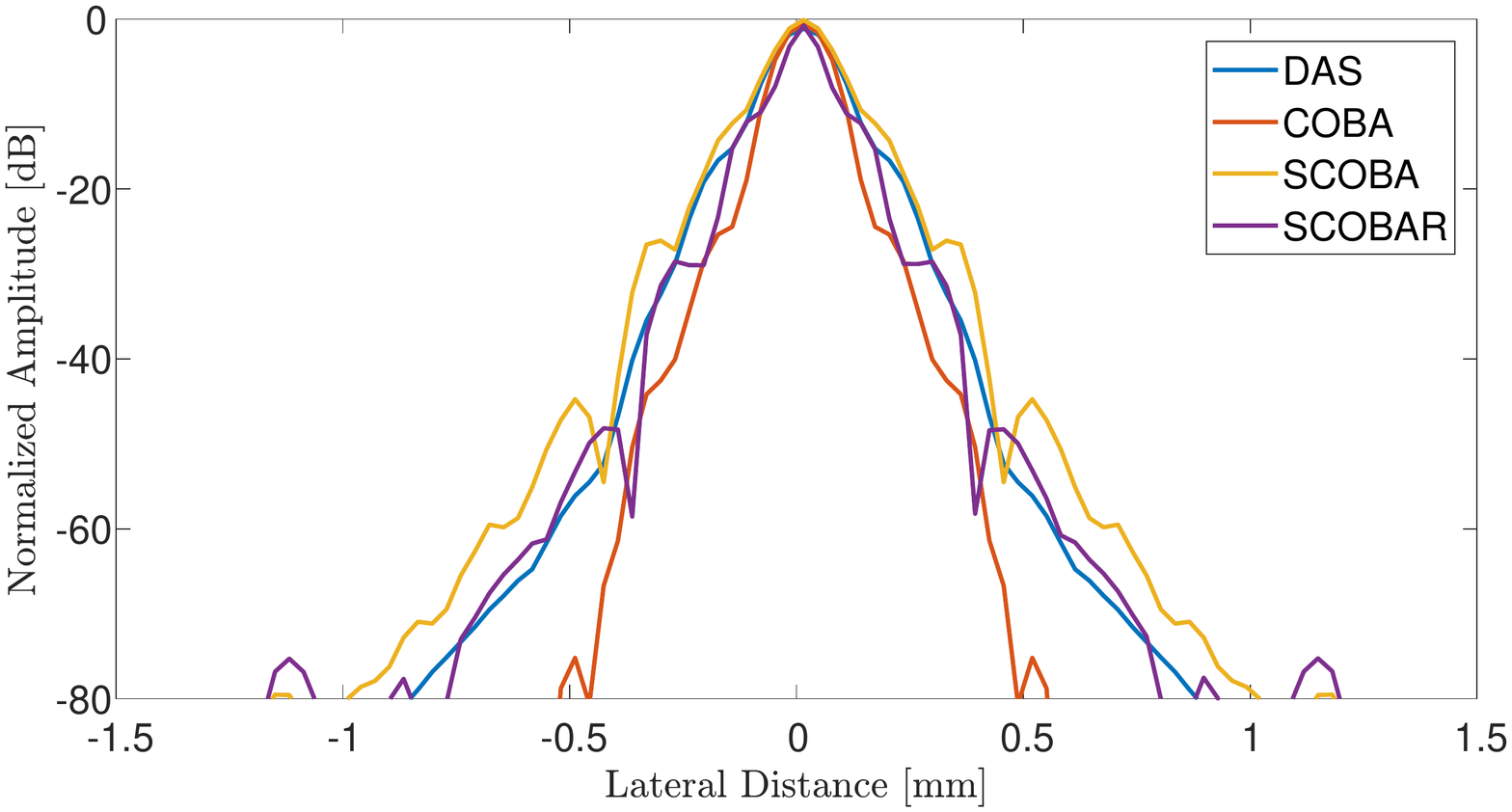}
  \caption{}
  \label{fig:sub2}
\end{subfigure}
\caption{(a) Axial profiles of all four methods at the center image line. (b) Lateral cross-sections of all four techniques at focal depth of 50 mm.}
\label{fig:resolution}
\end{figure*}

\begin{figure*}
 \centering
 \includegraphics[trim={7.5cm 2.5cm 6cm 4cm},clip, height = 11cm, width = 0.9\linewidth]{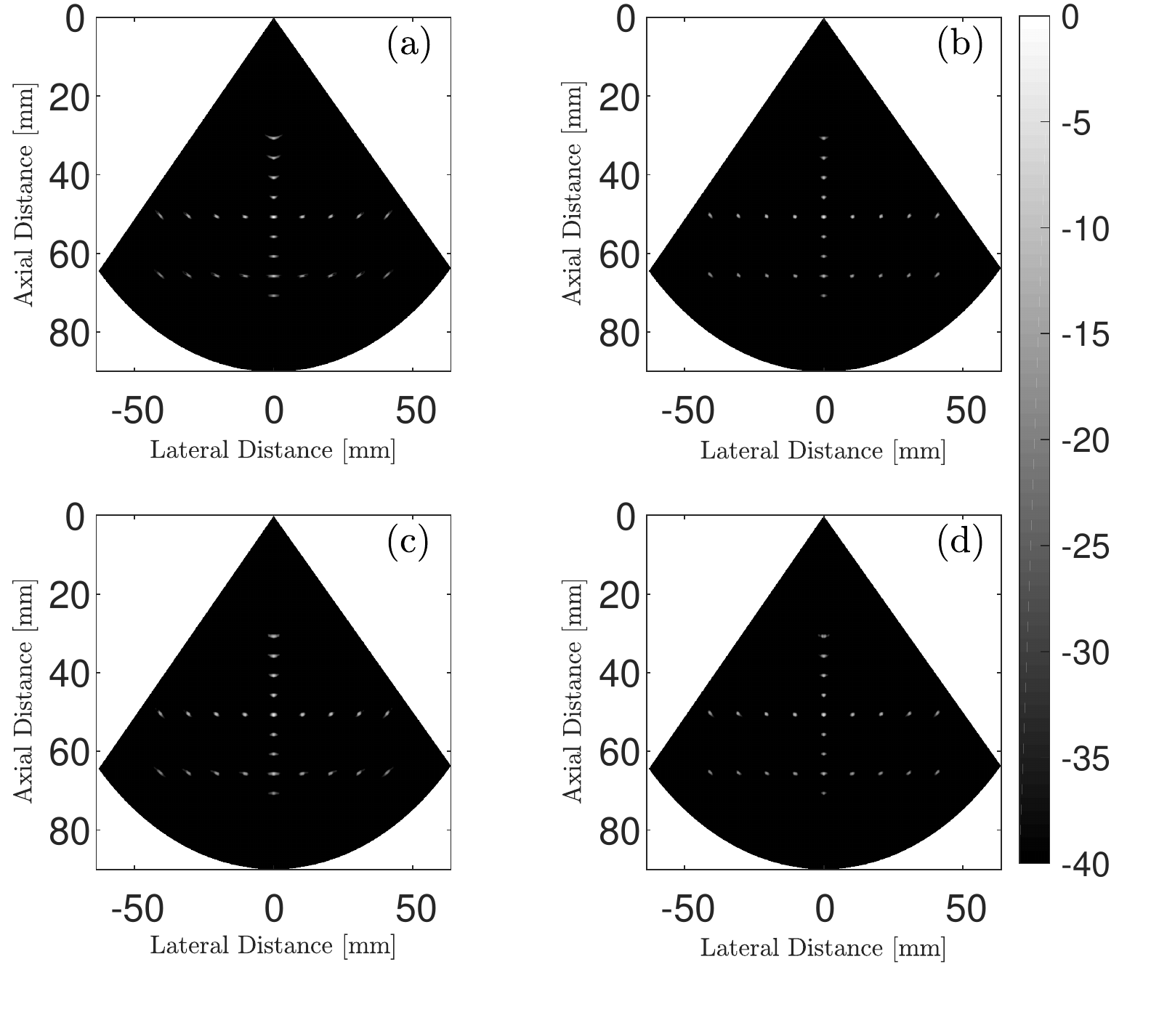}
 \caption{Images of simulated on-axis and off-axis point targets obtained by (a) DAS (127), (b) COBA (127), (c) SCOBA (29), (d) SCOBAR (43). 
%  The dynamic range for all the images is 40 dB.
 Number in brackets refers to the number of elements used.}
  \label{fig:offaxis}
 \end{figure*}

\subsubsection{Contrast}
For contrast evaluation, we use a simulated phantom of an anechoic cyst embedded in a speckle background. Figure \ref{fig:imcontrast} displays the images obtained with DAS, COBA, SCOBA and SCOBAR and provides a qualitative impression of the contrast achieved by each method. In addition, lateral cross-sections of the cyst at depth of 64 mm (dashed line in Fig. \ref{fig:imcontrast}) are presented in Fig. \ref{fig:contrast}, showing that SCOBA and DAS  have similar contrast, SCOBAR demonstrates an improvement over the latter and COBA achieves the best performance. 

A quantitative measure of contrast is the contrast
ratio (CR) \cite{lediju2011short} 
\begin{equation}
\text{CR}=20\log_{10}\left(\frac{\mu_{\text{cyst}}}{\mu_{bck}}\right)
\end{equation}
where $\mu_{\text{cyst}}$ and $\mu_{\text{bck}}$ are the mean image intensities, prior to log-compression, computed over small regions inside the cyst and in the surrounding background respectively. The regions selected are designated by a dashed circles in Fig.\,\ref{fig:imcontrast}. Consistent with previous results, the CR of DAS is -30.1 dB and of SCOBA is similar and equal to -30 dB. The CR of SCOBAR and COBA are -34 dB and -44 dB, receptively. These results emphasize the superiority of COBA and demonstrate that similar and improved performance to that of DAS can be obtained while using much fewer elements.    
\begin{figure*}[h]
 \centering
 \includegraphics[trim={3cm 3cm 1cm 4cm},clip, height = 9cm, width = 0.9\linewidth]{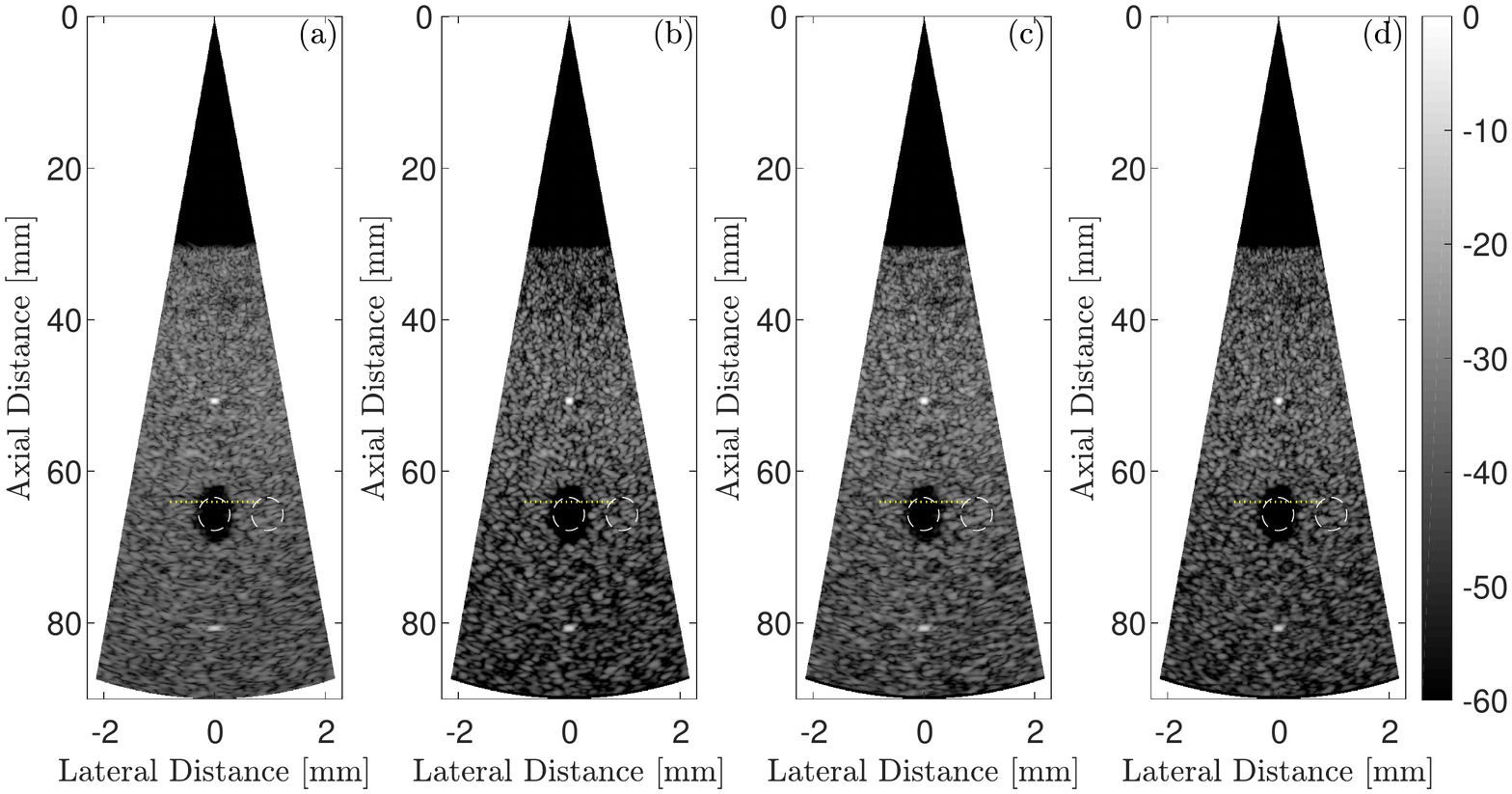}
 \caption{Images of simulated anechoic cyst phantom obtained by (a) DAS (127), (b) COBA (127), (c) SCOBA (29), (d) SCOBAR (43).
%  The dynamic range for all the images is 60 dB.
 The dashed line marks the lateral cross section presented in Fig. \ref{fig:contrast}. The dashed circles indicate the region used for computing the contrast ratios. Number in brackets refers to the number of elements used.}
  \label{fig:imcontrast}
 \end{figure*}

 \begin{figure}
 \centering
 \includegraphics[trim={3cm 3cm 1cm 3cm},clip, height = 4cm, width = 0.9\linewidth]{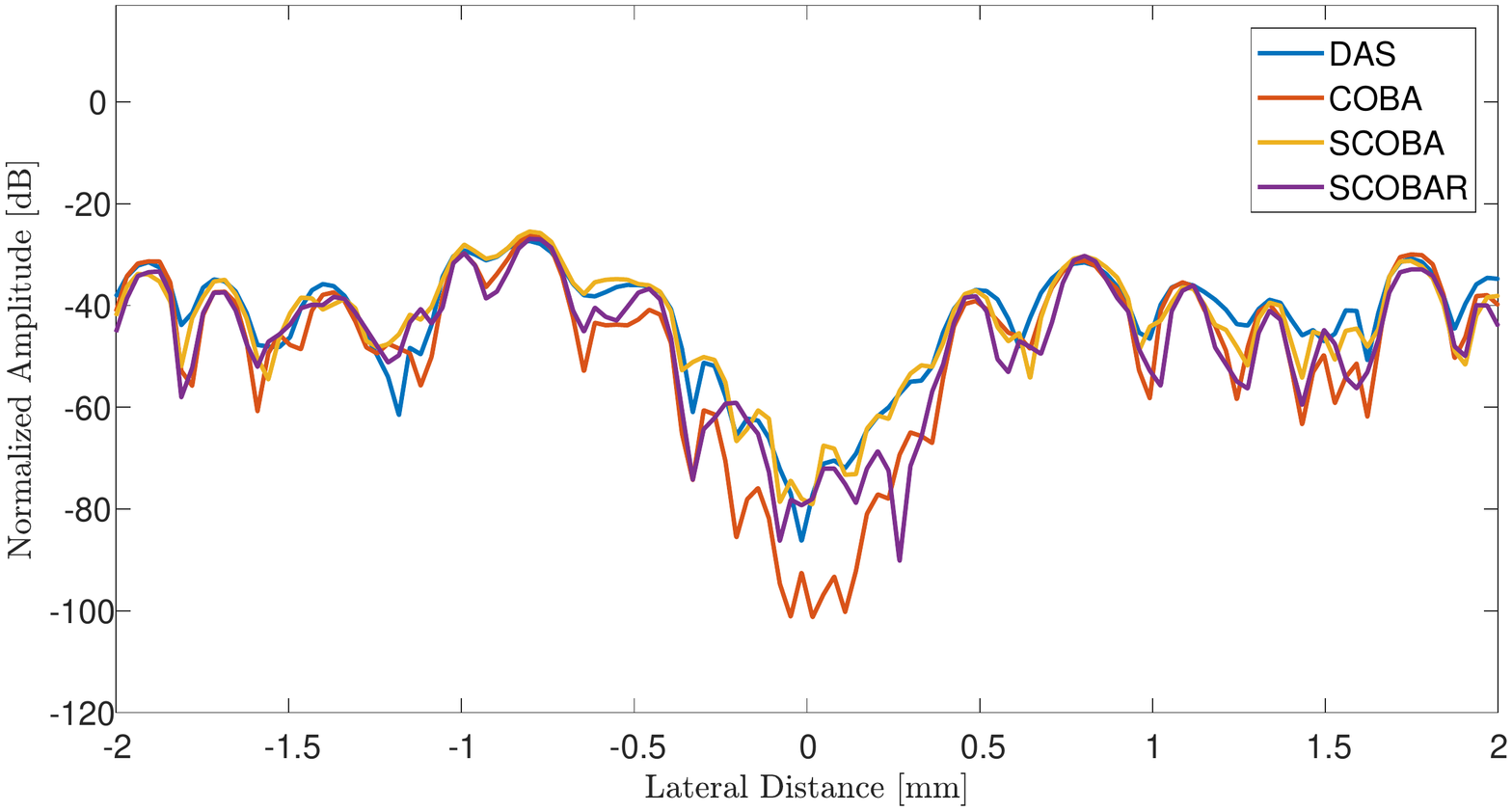}
 \caption{Lateral cross-sections of the cyst obtained by all four techniques.}
  \label{fig:contrast}
 \end{figure}

\subsection{Phantom Scans}
We next proceed to evaluate the proposed beamformers using experimental data. To that end, phantom data was acquired by a Verasonics Vantage 256 system. Tissue mimicking phantoms Gammex 403GSLE and 404GSLE were scanned by a 64-element phased array transducer P4-2v with a frequency response centered at 2.9 MHz and a sampling frequency of 11.9 MHz. The parameters for SCOBA and SCOBAR were chosen to be $A=4$ and $B=8$, resulting in 21 and 27 elements respectively. The results obtained from different phantom scans are presented in Figs\,\,\ref{fig:phantomcyst}, \ref{fig:phantomres} and \ref{fig:phantomlarge}, and include on-axis and off-axis targets, various cysts and resolution target groups. Zoom in on areas of cysts and resolution targets are shown in Figs\,\,\ref{fig:phantomzoomin} and \ref{fig:phantomzoomincyst} respectively. As can be seen, COBA exhibits an improvement over DAS in terms of contrast and resolution, SCOBA and SCOBAR achieve similar performance to DAS and COBA respectively, while using fewer elements.  
\begin{figure*}
 \centering
 \begin{minipage}{1\textwidth}
 \includegraphics[trim={2cm 8cm 2cm 8cm},clip, height = 4cm, width = 0.9\linewidth]{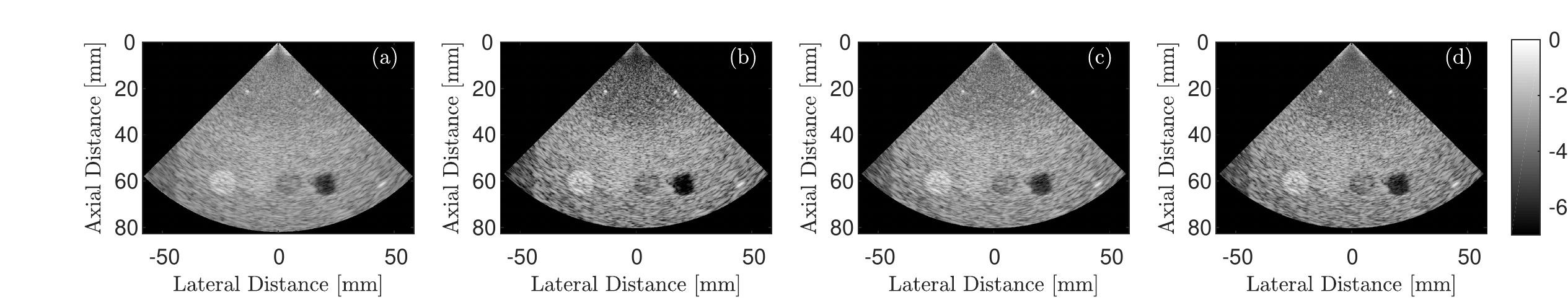}
 \caption{Images of GAMMEX 403GSLE which include pin and cystic targets obtained with (a) DAS (63), (b) COBA (63), (c) SCOBA (21), (d) SCOBAR (29). Number in brackets refers to the number of elements used.}
  \label{fig:phantomcyst}
  \par\vfill
 \includegraphics[trim={2cm 8cm 2cm 8cm},clip, height = 4cm, width = 0.9\linewidth]{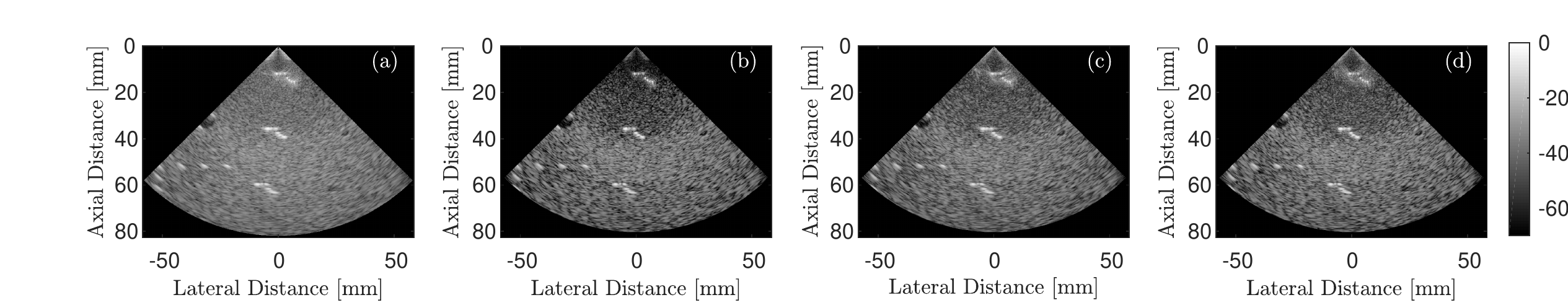}
 \caption{Images of GAMMEX 404GSLE which include pin and resolution targets obtained with (a) DAS (63), (b) COBA (63), (c) SCOBA (21), (d) SCOBAR (29). Number in brackets refers to the number of elements used.}
  \label{fig:phantomres}
   \includegraphics[trim={2cm 8cm 2cm 8cm},clip, height = 4cm, width = 0.9\linewidth]{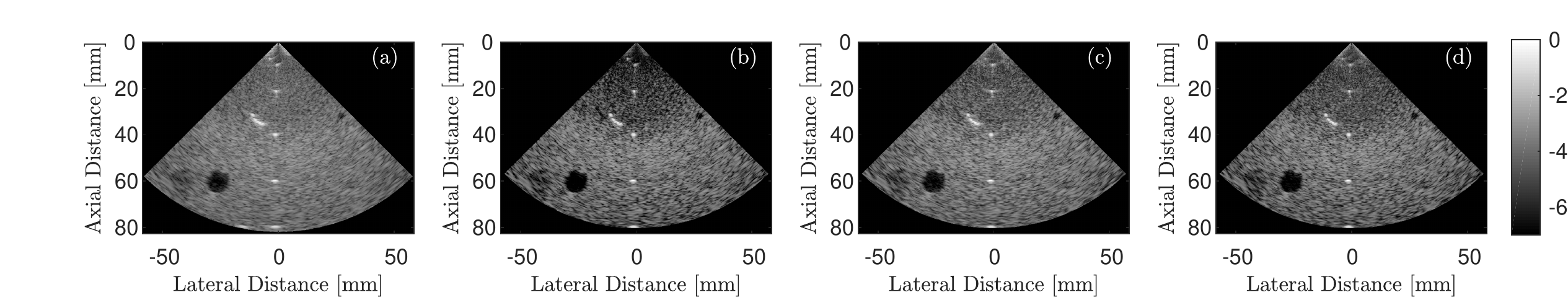}
 \caption{Images of GAMMEX 403GSLE which include resolution and cystic targets obtained with (a) DAS (63), (b) COBA (63), (c) SCOBA (21), (d) SCOBAR (29). Number in brackets refers to the number of elements used.}
  \label{fig:phantomlarge}
   \includegraphics[trim={2cm 8cm 2cm 8cm},clip, height = 4cm, width = 0.9\linewidth]{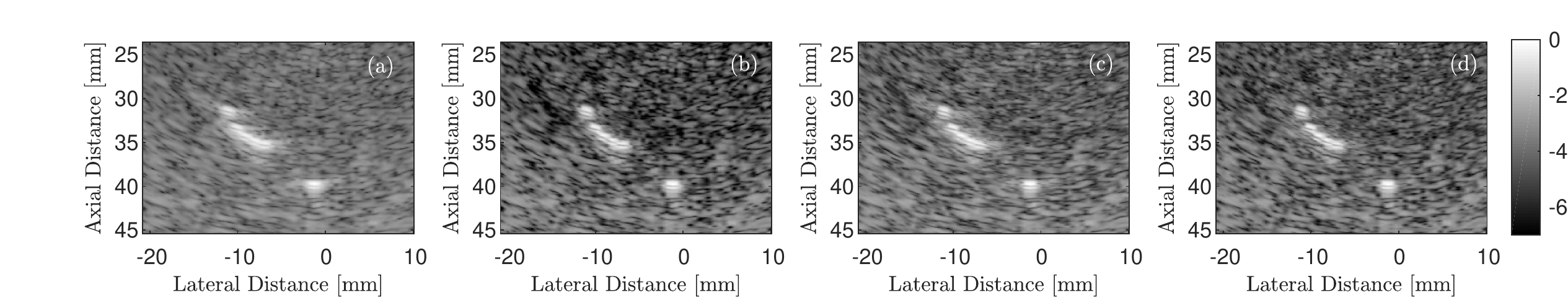}
 \caption{Zoom in on the resolution targets shown in Fig. \ref{fig:phantomlarge}.}
  \label{fig:phantomzoomin}
   \includegraphics[trim={2cm 8cm 2cm 8cm},clip, height = 4cm, width = 0.9\linewidth]{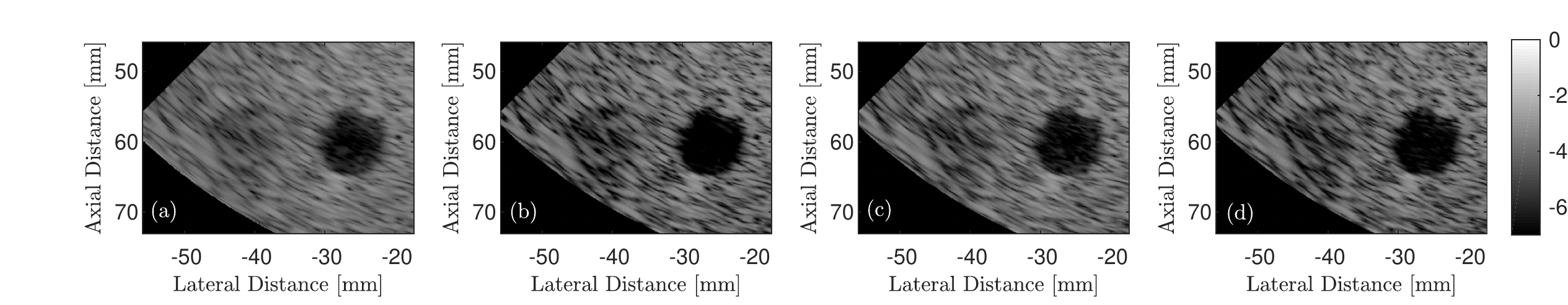}
 \caption{Zoom in on the cystic targets displayed in Fig. \ref{fig:phantomlarge}.}
  \label{fig:phantomzoomincyst}
  \end{minipage}%
 \end{figure*}

\subsection{In Vivo Acquisition}
Finally, we apply the proposed methods on \textit{in vivo} cardiac data. The acquisition
was performed with a GE breadboard ultrasonic scanner where 63 acquisition
channels were used. The radiated depth was 16 cm, the probe carrier frequency was 3.4 MHz and the system sampling frequency was 16 MHz. For COBA, SCOBA and SCOBAR a Hanning window-based high-pass filter was used (rather than a band-pass) with a cutoff frequency of 5 MHz, as shown in Figure \ref{fig:filterIV}. The parameters for SCOBA and SCOBAR were set to $A=4$ and $B=8$, leading to the minimal numbers of elements that can be obtained as stated in Theorems \ref{theo:scobamin} and \ref{theo:scobarmin}. Consequently, 21 and 27 elements out of 63 were used by SCOBA and SCOBAR respectively.

The results are presented in Fig. \ref{fig:invivo}. Clearly, COBA outperforms DAS in terms of image quality; the background noise is reduced and the anatomical structures are better highlighted. SCOBA achieves similar resolution as DAS, whereas, SCOBAR yields notable resolution improvement. Moreover, both the sparse beamformers obtain a low noise floor compared to DAS and thus the heart walls are better defined. These results validate that using the proposed techniques a reduction in the number of elements can be attained without compromising and even improving the image quality in comparison to standard DAS.  

\begin{figure*}
 \centering
 \includegraphics[trim={3cm 3cm 3cm 4cm},clip, height = 9cm, width = 0.9\linewidth]{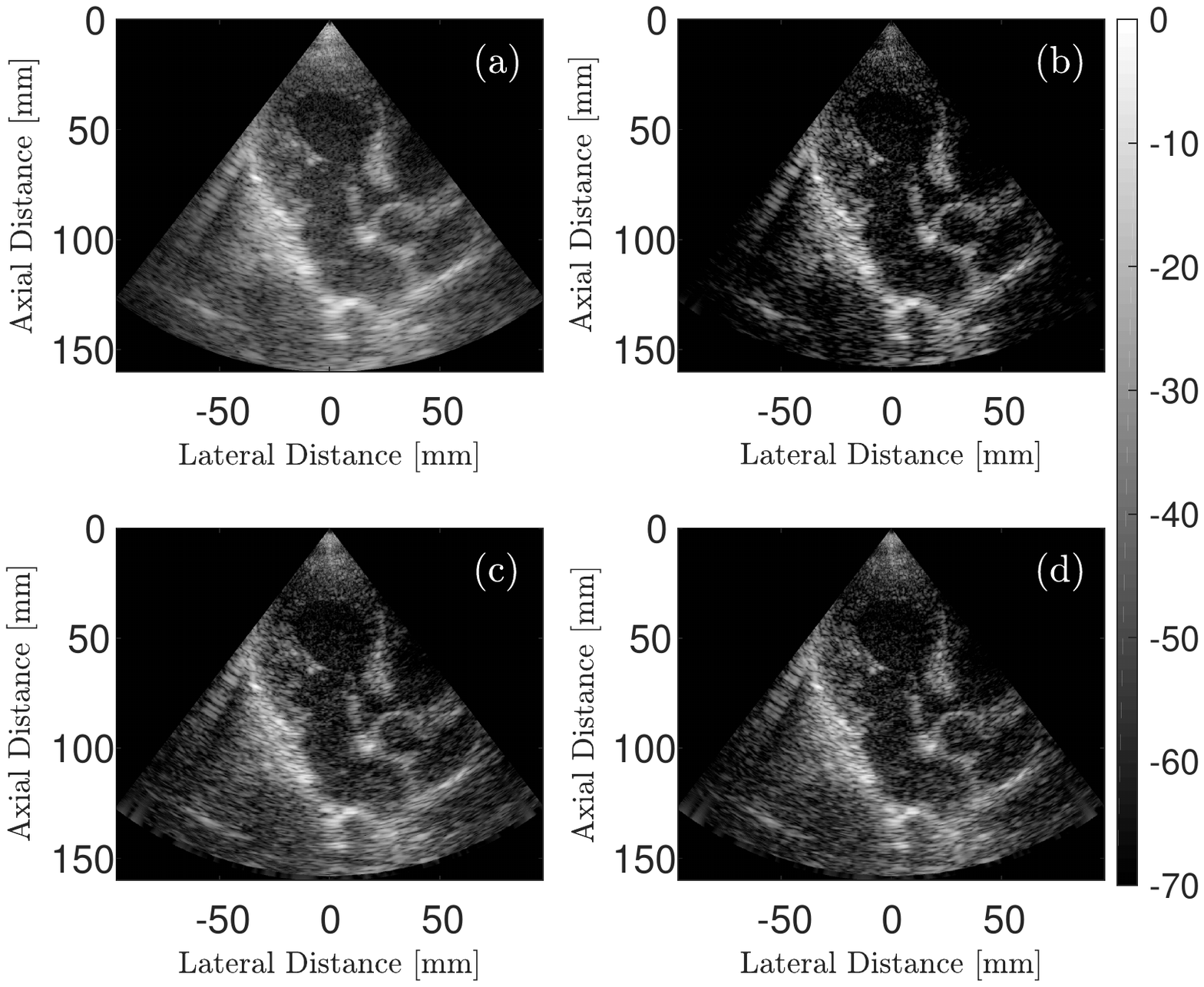}
 \caption{Cardiac images obtained with (a) DAS (63), (b) COBA (63), (c) SCOBA (21), (d) SCOBAR (29).
%  All the images are presented with a dynamic range of 70 dB.
 Number in brackets refers to the number of elements used.}
  \label{fig:invivo}
 \end{figure*}

\begin{figure}
 \centering
 \includegraphics[trim={3cm 3cm 3cm 3cm},clip,height = 3.5cm, width = 0.7\linewidth]{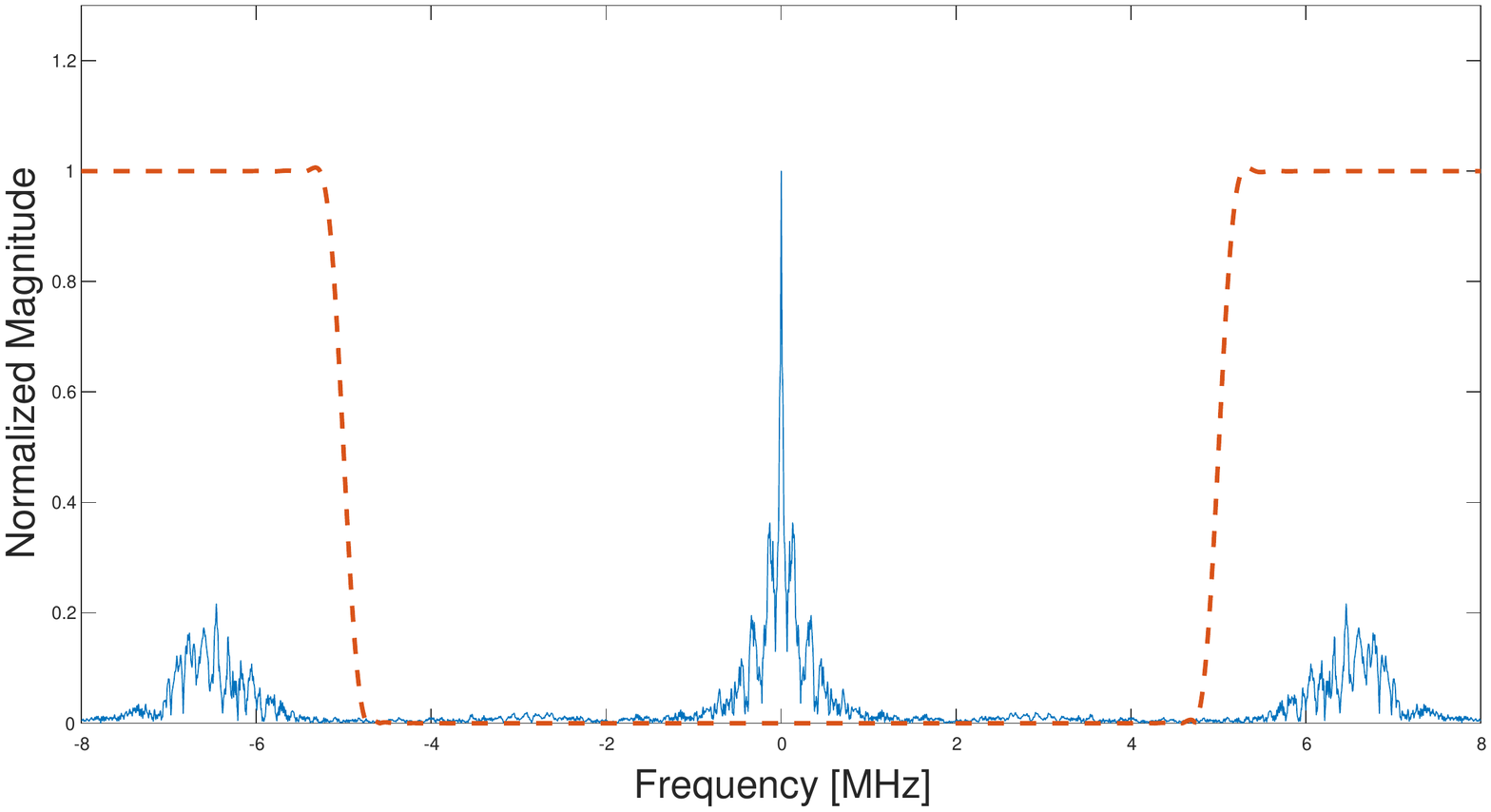}
 \caption{Fourier transforms of the impulse response of the Hanning-based high-pass filter (dashed line) and of the signal given by ($\ref{eq:bconv}$), which corresponds to the central in-vivo image line.}
  \label{fig:filterIV}
 \end{figure}

\section{Conclusion}
\label{sec:conclude}
In this paper we proposed three techniques for beamforming upon reception. First, we introduced a beamformer called COBA, which is based on convolution of the RF signals and is implemented efficiently using FFT. Then, we introduced the concept of sum co-array to analyze the beam pattern generated by COBA, showing it yields twofold enhancement in lateral resolution, compared to standard DAS, and provides contrast improvement. This was validated using qualitative and quantitative measurements in simulations which emphasized that COBA leads to an increase in resolution, contrast and noise suppression. In addition, an \textit{in vivo} scan was provided for visual assessment of the resulting image quality. 

Based on COBA and the sum co-array, we next presented two sparse beamformers, SCOBA and SCOBAR, which utilize a reduced number of elements. SCOBA requires much fewer elements without degrading image quality compared to DAS, whereas SCOBAR offers an improvement of resolution and contrast at the expense of a smaller, yet sizable, element reduction. The minimal number of elements in both algorithms is proportional to the square root of the number used with a full array. In addition, SCOBA may allow for a probe with a smaller physical aperture. The performance of SCOBA and SCOBAR was studied using both simulated and experimental data, verifying that only a small number of elements can be used while maintaining or improving the image quality compared to DAS. Images of \textit{in vivo} cardiac scans demonstrate that SCOBA and SCOBAR are suitable for clinical use. 

To conclude, the proposed methods provide a prominent improvement of contrast and lateral resolution in comparison with DAS. In addition, they allow for a significant element reduction while preserving or enhancing image quality. Thus, they enable the design of cheap, portable probes and low power ultrasound systems with a low computational load, paving the way to 3D imaging and wireless operation. 

\appendices
\section{Discrete Convolution}
\label{app:polyprod}
Consider two discrete sequences $a$ and $b$ of length $N+1$ and $M+1$ respectively. The discrete linear convolution of $a$ and $b$ is a sequence $c$ of length $L=N+M+1$ whose entries are given by
\begin{equation}
c_s = \sum_{i=0}^s a_{s-i}b_i,\quad s=0,1,...,L-1.
\end{equation}
Here $a$ and $b$ are zero padded to be of length $L$.

Let $f$ and $g$ be two polynomials defined by 
\begin{equation}
f(p)=\sum_{n=0}^N a_np^n,\quad g(p)=\sum_{m=0}^M b_mp^m.
\end{equation}
Their product is
\begin{equation}
h(p)\triangleq f(p)g(p)= \sum_{n=0}^N \sum_{m=0}^M  a_nb_mp^np^m.
\end{equation}
The latter can be viewed as a sum of single powers of $p$ by substituting $s=n+m$
\begin{equation}
h(p)= \sum_{s=0}^{L-1} \left(\sum_{(n,m):\, n+m=s}  a_nb_m\right) p^s.
\end{equation}
The coefficients $c_s$ of this polynomial are given by the inner summation which can be expressed as
\begin{equation}
c_s=\sum_{(n,m):\, n+m=s}  a_nb_m = \sum_{i=0}^s  a_{s-i}b_s,
\end{equation}  
where the second equality is obtained by zero-padding $a$ and $b$ to be of length $L$.
Thus, the coefficients of $h(p)$ are the linear convolution of $a$ and $b$.  

\section{Proof of Theorem \ref{theo:scobamin}}
\label{app:scobaproof}
We consider the equivalent problem given by (\ref{eq:scobaopt2}). It is clear from the constraints that $A$ and $B$ are both divisors of $N$ and we can express $B$ using $A$ as $B=\frac{N}{A}$. Without loss of generality, we assume that $A\leq B$ which leads to the following formulation
\begin{equation}
A^*=\underset{m\in D_1}{\arg\min}\quad m+\frac{N}{m}.
\end{equation}

Next, we define a function $g:[1,\sqrt{N}]\rightarrow\mathbb{R}^+$ over a continuous domain
\begin{equation*}
g(x)=x+\frac{N}{x}.
\end{equation*}
The function $g(x)$ is continuous and differentiable over the open
domain $(1,\sqrt{N})$. Its derivative is given by
\begin{equation*}
\frac{dg}{dx}=1-\frac{N}{x^2}< 0,
\end{equation*}
hence, $g(x)$ is monotonically decreasing. Using the fact that $D_1\subseteq [1,\sqrt{N}]$ and denoting $n=\max(D_1)$, it holds that 
\begin{equation*}
g(n)< g(m),\quad m\in D_1,\, m\neq n.
\end{equation*}
Therefore, the optimal solution is given by $A^*=n=\max(D_1)$ and $B^*=\min(D_2)$ accordingly. The solution for $B\leq A$ is established with the same arguments by interchanging the roles of $A$ and $B. \hfill \square$
\newline

Notice that when $N$ is a perfect square we have that $\max(D_1)=\min(D_2)$, leading to the single solution described earlier. In general, there are two optimal solutions, however, the solution in which $B\geq A$ is superior to the second one in terms of mutual coupling.

\section{Proof of Theorem \ref{theo:scobarmin}}
\label{app:scobarproof}

We consider  the equivalent problem (\ref{eq:scobaropt2}). Denoting $M=2A\in \mathbb{E}$, we rewrite it as
\begin{align}
\begin{split}
M^*,B^* =\underset{M,B\in\mathbb{N},\,B>1}{\arg\min}\quad &M+B \\
\text{subject to}\quad &MB=2N, \\
& M\in \mathbb{E}.
\end{split}
\label{eq:evenproblem}
\end{align}
Ignoring for a moment the last constraint, problem (\ref{eq:evenproblem}) is similar to (\ref{eq:scobaopt2}) with $2N$ replacing $N$. Hence, by similar arguments to those presented in the proof of Theorem \ref{theo:scobamin}, we have that
\begin{align}
\begin{split}
&M^*=\max(D_3)\text{ and } B^*=\min(D_4), \\
&M^*=\min(D_4)\text{ and } B^*=\max(D_3).
\label{eq:optimalmn}
\end{split}
\end{align}
Now, we enforce the constraint $M\in\mathbb{E}$. Since $\max(D_3)\min(D_4)=2N$ either $\max(D_3)$ or $\min(D_4)$ are even, or both, therefore, at least one of the optimal solutions in (\ref{eq:optimalmn}) is valid. Thus, taking into account that $A^*=M^*/2$ we get the optimal solutions presented for each one of the three cases. $\hfill \square$
\newline

Notice that when $2N$ is a perfect square we have that $\max(D_3)=\min(D_4)=\sqrt{2N}\in\mathbb{E}$ and the solution is $A=\sqrt{\frac{N}{2}},\,B=\sqrt{2N}$, as presented before.

% \section*{Acknowledgment}
\FloatBarrier
\bibliographystyle{IEEEtran}
\bibliography{IEEEabrv,REFS}

\end{document}